\documentclass[11pt]{article}
\usepackage[margin=1in]{geometry}
\usepackage{tcolorbox}
\usepackage{wrapfig}
\usepackage{breakurl}
\usepackage{bold-extra}
\usepackage{lipsum}
\usepackage{amsfonts}
\usepackage{graphicx}
\usepackage{stfloats}
\usepackage{enumitem}
\usepackage{xspace}
\usepackage{url}
\usepackage{comment}
\usepackage{caption}
\usepackage{epstopdf}
\usepackage{verbatim}
\usepackage{times}
\usepackage{amsmath, amsthm, amssymb}
\usepackage{color}
\usepackage{graphicx}
\usepackage{cite}
\usepackage{epsfig}
\usepackage{url}
\usepackage{xcolor}
\usepackage{xspace}
\usepackage{pgfgantt}
\usepackage{booktabs}
\usepackage{authblk}
\usepackage{algorithmic}
\usepackage{algorithm2e}
\usepackage{tabularx}
\usepackage{setspace}
\usepackage{amsmath}

\ifpdf
  \DeclareGraphicsExtensions{.eps,.pdf,.png,.jpg}
\else
  \DeclareGraphicsExtensions{.eps}
\fi

\usepackage{xcolor}
\usepackage{subcaption}
\usepackage{wrapfig}
\usepackage{array}
\usepackage{array}
\usepackage{tabulary}
\usepackage{ctable} 
\usepackage[normalem]{ulem} 

\newif\ifcomments
\commentstrue

\newcommand{\whp}{with high probability}

\newtheorem{theorem}{Theorem}
\newtheorem{corollary}{Corollary}
\newtheorem{lemma}{Lemma}

\newtheorem{definition}{Definition}

\newcommand{\AOneL}{c(\mbox{A1, \texttt{low}})\xspace}
\newcommand{\AOneH}{c(\mbox{A1, \texttt{high}})\xspace}
\newcommand{\ATwoL}{c(\mbox{A2, \texttt{low}})\xspace}
\newcommand{\ATwoH}{c(\mbox{A2, \texttt{high}})\xspace}

\newcommand{\POW}{PoW\xspace}

\newcommand{\AlgA}{\textsc{SybilControl\xspace}}
\newcommand{\AlgB}{\textsc{CCom}\xspace}

\newcommand{\Diffuse}{\textsc{Diffuse}\xspace}

\newcommand{\ppuzzle}{purge\ puzzle\xspace}
\newcommand{\epuzzle}{entrance\ puzzle\xspace}

\newcommand{\tstamp}{\mathcal{T}}

\newcommand{\Old}{\mathcal{S}_{\mbox{\tiny i}}}

\newcommand{\Iters}{\mathcal{I}}

\newcommand{\qcomm}{committee\xspace}
\newcommand{\cgoal}{Committee Invariant\xspace}
\newcommand{\sgoal}{Population Invariant\xspace}

\newcommand{\joinRate}{J^G\xspace}  
\newcommand{\iJRate}{J^G\xspace} 
\newcommand{\joinRateBad}{J^B\xspace}

\newcommand{\depRateTot}{D\xspace}  

\newcommand{\joinRateAll}{J^{\mbox{\tiny all}}\xspace}  

\newcommand{\elen}{\ell\xspace} 
\newcommand{\epochRate}{\rho\xspace}  

\newcommand{\estSet}{\tilde{S}\xspace} 
\newcommand{\estDur}{\tilde{\elen}\xspace}  

\newcommand{\advCost}{\mathcal{T}\xspace}
\newcommand{\advAveCost}{T\xspace}

\newcommand{\algGM}{\textsc{GMCom}\xspace}
\newcommand{\genID}{\textsc{GenID}\xspace}

\newcommand{\rateCur}{J_{\mbox{\tiny cur}}\xspace}
\newcommand{\setIDs}{\mathcal{S}\xspace}
\newcommand{\setGood}{\mathcal{G}\xspace}

\newcommand{\iterIDs}{\mathcal{S}\xspace}
\newcommand{\curIDs}{\mathcal{S}_{\mbox{\tiny cur}}\xspace}

\newcommand{\JoinEst}{\tilde{J}^{G}\xspace}
\newcommand{\bigconstant}{c(JE, \texttt{high})\xspace}
\newcommand{\smallconstant}{c(JE, \texttt{low})\xspace}

\newcommand{\advRate}{\ensuremath{J^B}\xspace}
\newcommand{\probDef}{\textsc{DefID}\xspace}

\newcommand{\defn}[1]{\textbf{\emph{#1}}}

\title{Resource-Competitive Sybil Defenses}
\date{}

\author[1]{Diksha Gupta}
\author[2]{Jared Saia\thanks{This work is supported by the National Science Foundation grants CNS-1318880 and CCF-1320994.}}
\author[3]{Maxwell Young\thanks{This work is supported by the National Science Foundation grant CCF 1613772 and by a research gift from C Spire.}}
\affil[1]{\small Dept. of Computer Science, University of New Mexico, NM, USA\hspace{6cm} \mbox{\texttt{dgupta@unm.edu}}}
\affil[2]{\small Dept. of Computer Science, University of New Mexico, NM, USA\hspace{6cm} \mbox{\texttt{saia@cs.unm.edu}}}
\affil[3]{\small Computer Science and Engineering Dept., Mississippi State University, MS, USA\hspace{6cm} \texttt{myoung@cse.msstate.edu}}

\begin{document}


\maketitle

\begin{abstract}
Proof-of-work (\POW) is an algorithmic tool used to secure networks by imposing a computational cost on participating devices. Unfortunately, traditional \POW schemes require that correct devices perform significant computational work in perpetuity, even when the system is not under attack.

We address this issue by designing general \POW protocols that ensure two properties. First, the fraction of identities in the system that are controlled by an attacker is a minority. Second, the computational cost of our protocol is comparable to the cost of an attacker.  In particular, we present an efficient algorithm, \algGM, which guarantees that the average computational cost to the good ID per unit time is $O(J + \sqrt{T(J+1)})$, where $J$ is the average number of joins by the good IDs and $T$ is the average computational spending of the adversary. Additionally, we discuss a precursor to this algorithm, \AlgB, which guarantees an average computational cost to good IDs per unit time of $O(J + T)$.

We prove a lower bound showing that \algGM's spending rate is asymptotically optimal among a large family of algorithms.  Finally, we provide empirical evidence that our algorithms can be significantly more efficient than previous defenses under various attack scenarios. 

\end{abstract}

\section{Introduction}\label{sec:intro}

Twenty-five years after its introduction by  Dwork and Naor ~\cite{dwork:pricing}, \defn{proof-of-work (\POW)} is enjoying a research renaissance.  Succinctly, \POW is an economic tool to prevent abuse of a system by requiring participants to solve computational puzzles in order to access system resources or participate in group decision making.  In recent years, \POW~is playing a critical role in cryptocurrencies such as Bitcoin~\cite{nakamoto:bitcoin}, and other blockchain technologies~\cite{litecoin,blockstack,chain1,iota,ethereum,dashcoin,primecoin}. Yet, despite success with Bitcoin and other cryptocurrencies, \POW~has not found a widespread application in mitigating malicious behavior in networks. This is despite numerous proposals for PoW-based defenses~\cite{parno2007portcullis,wang:defending,kaiser:kapow,green:reconstructing,feng:design,waters:new,martinovic:wireless,borisov:computational,li:sybilcontrol}.


\medskip

\noindent{\bf A Barrier to Widespread Use.} A major impediment to the widespread use of~\POW~is ``the work''. In particular, current \POW approaches have significant computational overhead, given that puzzles must always be solved, {\it even when the system is not under attack}.  This non-stop resource burning translates into a substantial energy --- and, ultimately, a monetary -- cost~\cite{economistBC,coindesk,arstechnica}.\footnote{For example, in 2018, the Economist calculated that Bitcoin consumes at least $22$ terawatt-hours of electricity per year, or enough to power Ireland~\cite{economistBC}.}  Consequently,  \POW approaches are currently used primarily in applications where participants have a financial incentive to continually perform work, such as cryptocurrencies.  This is a barrier to wide-spread use of a technique that has the potential to fundamentally advance the field of cybersecurity.  

In light of this, we seek to reduce the cost of \POW systems and focus on the following question: \textbf{Can we design \POW systems where resource costs are low in the absence of attack, and grow slowly with the effort expended by an attacker?}

In this paper, we design and analyze algorithms that answers this question in the affirmative. {We begin by discussing a straightforward algorithm, \defn{\AlgB} in Section \ref{sec:server1}, followed by some minor enhancements to this approach in Section~\ref{sec:enhancements}. Next, we present an improved algorithm, \defn{\algGM} (Section \ref{sec:gmcom}) and analyze it (Section~\ref{s:anal-gmcomm}).  Then we prove a matching lower bound for this class of algorithms in Section \ref{sec:lower}.} We present  empirical results which complement our theoretical work in Section \ref{sec:experiments}. Finally, we conclude and offer directions for future research in Section \ref{sec:future}.


\subsection{Our Model}\label{sec:model-main}

Our system consists of virtual \defn{identifiers (IDs)}, and an attacker (or \defn{adversary}).  Each ID is either good or bad.   Each \defn{good} ID follows our algorithm, and all \defn{bad} IDs are controlled by the adversary. 

\subsubsection{Communication}\label{sec:com} All communication among good IDs occurs through a broadcast primitive, denoted by~{\bf \Diffuse}, which  allows  a  good  ID  to  send  a  value  to all other good IDs within a known and bounded amount of time, despite the presence of an adversary. We assume that when a message is diffused in the network, it is not possible to determine which ID initiated the diffusion of that message. Such a primitive is a standard assumption in \POW schemes~\cite{Garay2015,bitcoinwiki,GiladHMVZ17,Luu:2016}; see~\cite{miller:discovering} for empirical justification.  We assume that each message originating at a good ID is signed by the ID's private key. 

A \defn{round} is the amount of time it takes to solve our easiest computational puzzle plus the time to communicate the solution to the rest of the network via \Diffuse.  We assume that a constant number of rounds is sufficient to perform Byzantine consensus (see Section~\ref{sec:our-problems} for definition of Byzantine consensus and an algorithm to solve it). 

As is standard, we assume all IDs are synchronized. For simplicity, we initially assume that the time to diffuse a message is small in comparison to the time to solve computational puzzles.\footnote{A recent study of the Bitcoin network shows that the communication latency is 12 seconds in contrast to the 10 minutes block interval~\cite{croman2016scaling}, which motivates our assumption of computational latency dominating the communication latency in our system.} However, we later relax this assumption to account for a network latency which is upper bounded by $\Delta$ in Section \ref{sec:bounded-latency}. We pessimistically assume that the adversary can send messages to any ID at will, and that it can read the messages diffused by good IDs before sending its own.

\subsubsection{Puzzles} \label{sec:puzzle} We assume a source of computational puzzles of varying difficulty, whose solutions cannot be stolen or pre-computed.  This is a common assumption in \POW systems~\cite{nakamoto:bitcoin, li:sybilcontrol, andrychowicz2015pow}. We now provide an overview of the standard way in which this assumption is achieved; more details on the puzzle construction are provided in Section~\ref{subsec:puzzle-construction}.

All IDs have access to a hash function, {\boldmath{$h$}}, about which we make the  \emph{random oracle assumption}~\cite{bellare1993random,koblitz2015random}.  Succinctly, this assumption is that when first computed on an input, $x$, $h(x)$ is selected independently and uniformly at random from the output domain, and that on subsequent computations of $h(x)$ the same output value is always returned.  We assume that both the input and output domains are the real numbers between $0$ and $1$.  In practice, $h$ may be a cryptographic hash function, such as SHA-2~\cite{sha2}, with inputs and outputs of sufficiently large bit lengths. 

Solving a puzzle requires that an ID find an input $x$ such that $h(x)$ is less than some threshold.   The input found is the \defn{puzzle solution}. Decreasing this threshold value will increase the difficulty, since one must compute the hash function on more inputs to find an output that is sufficiently small.   

Finally, we must be able to address an adversary who attempts  (1)  to falsely claim puzzle solutions computed and transmitted by good IDs, and (2) to pre-compute solutions to puzzles. These details are deferred until Section~\ref{subsec:puzzle-construction} when we describe our algorithm. 

\subsubsection{Adversary}\label{sec:adv} A single adversary controls all bad IDs.  This pessimistically represents perfect collusion and coordination by the bad IDs. Bad IDs may arbitrarily deviate from our protocol, including sending incorrect or spurious messages. This type of adversary encapsulates the Sybil attack~\cite{douceur02sybil} (see Section~\ref{sec:related-work}).  

The adversary controls an $\alpha$-fraction of computational power, where $\alpha>0$ is a small constant.  Thus, in a single round where all IDs are solving puzzles, the adversary can solve an $\alpha$-fraction of the puzzles. This assumption  is standard in past \POW literature~\cite{nakamoto:bitcoin, andrychowicz2015pow, walfish2010ddos, GiladHMVZ17}.

Our algorithms employ digital signatures, but do not require any public key infrastructure. Further, we assume the adversary knows our algorithm, but does not know the private random bits of any good ID.  

\subsection{System Initialization}\label{sec:initialization} 
We assume an \defn{initialization phase} over which initial IDs join the system. At the end of the initialization phase, a heavy-weight protocol (\genID, described in Section~\ref{sec:genID}) is run, ensuring that (1) all IDs learn an initial system membership, $S_0$; and (2) no more than an $\alpha$-fraction of the membership of $S_0$ is bad.  The initial estimate of the join rate of good IDs is the size of $S_0$ divided by the time taken for the initialization phase.

\subsubsection{Joins and Departures}\label{sec:join} The system is dynamic with IDs joining and departing over time, subject to the constraint that at most a constant fraction of the good IDs can join or depart in any round.  Maintaining performance guarantees amidst many nodes joining and leaving a system is challenging in decentralized systems~\cite{Augustine:2013:SSD:2486159.2486170,7354403,Augustine2015,Augustine:2012,augustine:fast}.\footnote{Informally this process of nodes joining and departing over time is often referred to as ``churn".} For \AlgB we pessimistically assume that all join and departure events are scheduled in a worst-case fashion by the adversary.  However, we assume that departing good IDs are selected uniformly at random from the set of all good IDs in the population.

After the initialization phase, we assume that all good IDs announce their departure to the network.  In practice, this assumption could be relaxed through the use of \defn{heartbeat messages} that are periodically sent out to indicate that an ID is still in the network.

The minimum number of good IDs in the system at any point is assumed to be at least {\boldmath{$n_0$}}.  Our goal is to provide security and performance guarantees  for {\boldmath{$O(n_0^{\gamma})$} } joins and departures of IDs, for any desired constant {\boldmath{$\gamma$}} $ \geq 1$. In other words, the guarantees on our system hold with high probability (w.h.p.)\footnote{With probability at least $1-n_0^{-c}$ for any desired $c\geq 1$.} over this polynomial number of dynamic events. This implies security despite a system size that may vary wildly over time, increasing and decreasing polynomially in $n_0$ above a minimum number of $n_0$ good IDs. 

We assume that at most an {\boldmath{$\epsilon_a$}}-fraction of good IDs may join in a single round, where  $\epsilon_a>0$ is a small constant. Similarly, we assume an {\boldmath{$\epsilon_d$}}-fraction of good IDs may depart in a single round, where $\epsilon_d>0$ are small constants that depends on the constant number of rounds required to solve Byzantine consensus.


\subsection{Problem Definition}\label{sec:our-problems}

We define the {\defn {Defend-ID}} (\textbf{\textsc{\probDef}}) problem which consists of maintaining the following two security invariants.\medskip

\noindent{\defn{\sgoal}{\bf:}} Ensure that the fraction of bad IDs in the entire system is less than $1/2$.
\medskip

We use $1/2$ for ease of presentation, but there is nothing special about this fraction. We can adjust this to be any fraction less than $3\alpha$ by adjusting our algorithm as discussed in Section \ref{sec:popgoalmod}. Achieving the \sgoal will help ensure that system resources -- such as jobs executed on a server, or bandwidth obtained from a wireless access point -- consumed by the bad IDs are (roughly) proportional to the adversary's computational power.  Possible application domains include: content-sharing overlays~\cite{falkner:profiling,steiner:global}; and open cloud services~\cite{Mohaisen:2013:TDC:2484313.2484332,Anderson:2004:BSP:1032646.1033223,Chandra:2009:NUD:1855533.1855535,WeissmanSGRNC11}. 

\medskip  

\noindent{\defn{\cgoal}{\bf:}} Ensure there always exists a committee that (i) is known to all good IDs; (ii) is of \emph{scalable} size i.e., of size $\Theta(\log n_0)$; and (iii) contains less than a $1/2$  fraction of bad IDs.

\medskip

Again, there is nothing special about $1/2$, and this can be adjusted to a smaller fraction of bad IDs (see Section~\ref{lem:maj_comm-GM}). Achieving the \cgoal ensures that the committee can solve the \defn{Byzantine consensus}~problem~\cite{lamport1982byzantine} in a scalable manner.  In this problem, each good ID has an initial input bit. The goal is for (i) all good IDs to decide on the same bit; and (ii) this bit to equal the input bit of at least one good ID.  
Byzantine consensus enables participants in a distributed network to reach agreement on a decision, even in the presence of a malicious minority.  Thus, it is a fundamental building block for many cryptocurrencies~\cite{BonneauMCNKF15,eyal2016bitcoin,cryptoeprint:2015:521,GiladHMVZ17}; trustworthy computing \cite{castro1998practical,castro2002practical,cachin:secure,kotla2007zyzzyva,clement-making,1529992}; peer-to-peer networks~\cite{oceanweb,adya:farsite}; and databases~\cite{GPD,preguica2008byzantium,zhao2007byzantine}.

Establishing Byzantine consensus via the use of committees is a common  approach; for examples, see~\cite{GiladHMVZ17,Luu:2016,KSSV,pass2016hybrid}.  There are algorithms that enable a set of $x$ nodes to solve Byzantine consensus in constant time and sending $O(x^2)$ messages, when communication is synchronous, nodes have digital signatures, and the fraction of Byzantine nodes is strictly less than $1/2$.  We make use of a recent algorithm by Abraham et al.~\cite{abraham2018synchronous}.

\medskip


\subsection{Results for \AlgB}\label{sec:main}

We measure \defn{computational cost} as the effort required to solve computational puzzles (see Section~\ref{sec:puzzle} for details). 

\smallskip
We begin by presenting our results for \AlgB (see Section \ref{sec:server1}). The \defn{system lifetime} consists of a number of ID joins and departures that is polynomial in $n_0$. Let  {\boldmath{$\advAveCost$}} be the  \defn{adversarial spending rate}, which is the cost to the adversary for solving puzzles over the system lifetime.  Let {\boldmath{$\joinRate$}} denote the \defn{good ID join rate}, which is the number of good IDs that join during the system lifetime. Finally, the \defn{algorithmic spending rate}  is the total cost to the good IDs for solving puzzles  over the system lifetime.

\begin{theorem}\label{thm:main1}
For $\alpha < 1/6$, with error probability polynomially small in $n_0$ over the system lifetime,  \AlgB solves \probDef with an algorithmic spending rate of $O(T + \joinRate)$.
\end{theorem}

Note that the computational cost incurred by the good IDs grows slowly with the cost incurred by the adversary. When there is no attack on the system, the costs are low and solely a function of the number of good IDs; there is no excessive overhead. But as the adversary spends more to attack the system, the costs required to keep the system secure grow commensurately with the adversary's costs. 

\subsection{\algGM Join and Departure Assumptions for Good IDs} \label{sec:modelGMCom}

Following are additional assumptions needed to achieve our main result for \algGM.  Time is divided into disjoint intervals called \defn{epochs}, defined as follows.

\begin{definition}\label{def:epoch}
Let {\boldmath{$\setGood_{i}$}} be the set of good IDs in the system at the end of epoch $i$ and let epoch $1$ begin at system creation. Then, epoch $i$ is defined as the shortest amount of time until $|\setGood_{i} - \setGood_{i-1}| \geq (3/4)|\setGood_{i}|$.
\end{definition}

Let {\boldmath{$\elen_i$}} denote the number of seconds in epoch $i$. Let {\boldmath{$\epochRate_i$}} be the join rate of good IDs in epoch $i$; that is, the number of good IDs that join in epoch $i$ divided by $\elen_i$.  

We make the assumptions A1 and A2 stated below. For ease of exposition, we use notation $c(\texttt{A}, \texttt{low})$ to indicate a positive constant at most $1$ used to lower bound a quantity in assumption \texttt{A}. Similarly, the constant $c(\texttt{A}, \texttt{high})\geq 1$ is used to upper bound a quantity in assumption \texttt{A}. By parameterizing our analysis later on with these constants, our results generalize to a variety systems where different bounds may be appropriate.

\medskip

\begin{itemize}[leftmargin=10pt]
\item{\bf A1.} For all $i > 1$, $\AOneL \epochRate_{i-1}  \leq \epochRate_{i} \leq \AOneH \epochRate_{i-1}$.\smallskip
\item{\bf A2.} For any period of time within an epoch that contains at least $2$ good join events, the good join rate during that period is between $\ATwoL\epochRate_{i}$ and $\ATwoH \epochRate_{i}$. \end{itemize} 
\medskip

Informally, assumption A1 implies that the rate at which good IDs join is fairly ``smooth'',  A2 implies that the ``rate of change" in the set of good IDs is at least $\Omega(J_i)$ over the course of an epoch, and assumption A2 prohibits good IDs  from being too ``bunched up" during a small interval of time.

In practice, the performance of our algorithms will depend on these constants, which may differ between various network scenarios. For this reason, we parameterize our analysis of \algGM (Sections~\ref{sec:estimating} and~\ref{sec:cost-analysis-gm}).  Later, in Section~\ref{s:JandLAssum}, we provide empirical results using data from real-world systems. 


\subsection{Results for \algGM}

Next, we present results for \algGM, presented in Section \ref{sec:gmcom}, which ensures that the algorithmic spending rate is asymptotically less than the adversarial spending rate in the presence of an attack, whereas in the absence of an attack it is commensurate with the join rate of good IDs.

\begin{theorem}\label{thm:main-upper}
For $\alpha \leq 1/18$, with error probability polynomially small in $n_0$ over the system lifetime, \algGM solves \probDef with an algorithmic spending rate of $O(\sqrt{\advAveCost (\joinRate+1)} + \joinRate)$.
\end{theorem}

This result is complemented by the following lower bound.  Define a \defn{purge-based} algorithm to be any algorithm where
 (1) IDs pay a cost of $\Omega(1)$ to join; and (2) after a constant fraction of the population changes, all IDs must pay $\Omega(1)$ to remain in the system; otherwise, they are purged.  

The execution of our algorithms,  \AlgB and \algGM,~consists of contiguous sequences of rounds called \defn{iterations}. An iteration that starts with round $s$ ends when the number of join and departure events since round $s$ exceeds a constant fraction of the number of IDs present at the start of round $s$. 

\begin{theorem}\label{thm:main-lower}
For any purge-based algorithm, there is an adversarial strategy ensuring the following for any iteration.  The algorithmic spending rate is $\Omega(\sqrt{\advAveCost\,\joinRate} + \joinRate)$, where $\joinRate$ is the good ID join rate and $\advAveCost$ is the algorithmic spending rate, over the iteration. 
\end{theorem}

\subsubsection{A Corollary for Theorem~\ref{thm:main-upper}}
The following Corollary of Theorem~\ref{thm:main-upper} shows that the spending rate for the algorithm remains small, even when focusing on just a subset of iterations.  To understand why this is important, consider a long-lived system which suffers a single, significant attack for a small number of iterations, after which there are no more attacks. The cost of any defense may be small when amortized over the lifetime of the system, but this does not give a useful guarantee on performance during the time of attack.

Let $\Iters$ be any subset of iterations that for integers $x$ and $y$, $1 \leq x \leq y$, contains every iteration with index between $x$ and $y$ inclusive, and let $L$ be the total length of time of those iterations.  Let $\delta(\Iters)$ be $|S_{x} - S_y|$; and let $\Delta(\Iters)$ be  $\delta(\Iters)$ divided by the length of $\Iters$.

The \defn{adversarial spending rate}, {\boldmath{$\advAveCost_{\mathcal{I}}$}},  is the cost to the adversary for solving puzzles whose solutions are used in any iteration of $\mathcal{I}$ divided by $L$. The \defn{good ID join rate}, {\boldmath{$\joinRate_{\mathcal{I}}$}}, is the number of good IDs that join over the iterations in $\mathcal{I}$ divided by $L$.  Next, the \defn{algorithmic spending rate}  is the total cost to the good IDs for solving puzzles whose solutions are used in any iteration of $\mathcal{I}$ divided by $L$. Then, we state our result formally as:

\begin{corollary}\label{cor:main-upper}
For $\alpha \leq 1/18$, with error probability polynomially small in $n_0$ over the system lifetime, and for any subset of contiguous iterations $\mathcal{I}$, \algGM has an algorithmic spending rate of: 
   $$O\left(\sqrt{\advAveCost_{\mathcal{I}}\,(\joinRate_{\Iters}+1)} + \Delta(\Iters)  + \iJRate_{\Iters}\right)$$. 
\end{corollary}

\subsection{Empirical Results}  In Section~\ref{sec:experiments}, we empirically evaluate our algorithms and compare their performance to the state of the art. For a variety of networks, our algorithms significantly reduce computational cost. Additionally, we describe heuristics that further improve the performance of our basic algorithms.


\section{Related Work} \label{sec:related-work}

\noindent\textbf{\AlgB and \algGM.} Preliminary results on \AlgB~\cite{pow-without} and \algGM~\cite{Gupta_Saia_Young_2019} appeared previously. In this full version of these results, we make several additional contributions. First, improvements to the design of both algorithms have allowed us to provide a more streamlined presentation, and reduce the complexity of our arguments; the accompanying proofs of correctness and performance are now included in full.   Second, the specification of \algGM contains a significant modification for estimating the join rate of good IDs that corrects an error in our preliminary version~\cite{Gupta_Saia_Young_2019}; the full correctness proofs for our new estimate are now provided  in Section~\ref{sec:estimating}.  Third, for \algGM, we are able to remove an assumption on the behavior of good IDs that was previously required in our preliminary results~\cite{Gupta_Saia_Young_2019}. Finally,  in Section~\ref{sec:experiments}, we present additional experimental results that align with our theoretical guarantees. 


\subsection{Sybil Attack and Defense}

\noindent\textbf{The Sybil attack.} Our result defends against the Sybil attack~\cite{douceur02sybil}.  There is large body of literature on mitigating the Sybil attack (for example, see surveys~\cite{newsome:sybil,mohaisen:sybil,john:soft,Dinger:2006}), and additional work documenting real-world Sybil attacks~\cite{bitcoin-sybil,6503215,Yang:2011:USN:2068816.2068841,Tran:2009:SOC:1558977.1558979}.  {\it Critically, none of these prior results ensure that  computational costs for good IDs grow slowly with the cost incurred by the adversary.}  We note that our algorithms can be characterized as \defn{resource-competitive}~\cite{gilbert:making,gilbert:near,king:conflict,bender:how,ICALP15,daniICJournal17,aggarwal2016secure,gilbert:resource,Bender:2015:RA:2818936.2818949,zamani2017torbricks}, in the sense that the algorithmic cost grows slowly with that of an attacker. However, for simplicity, we omit a discussion of resource-competitive algorithms since it is not critical to understanding our results.  We refer interested readers to the survey~\cite{Bender:2015:RA:2818936.2818949}.

We note that obtaining large numbers of machines is expensive: computing power costs money, whether obtained via Amazon AWS~\cite{amazon-aws} or a botnet rental~\cite{anderson2013measuring}.  This motivates our desire to consider the attacker's cost when designing our algorithm's defense.

\medskip

\noindent\textit{\bf \genID and System Initialization.} In the \genID problem, there is a set of good IDs, and an adversary that wants to impersonate many bad IDs. The adversary has an $\alpha-$fraction of computational power. All good IDs must agree on a set of IDs that contains (1) all good IDs, and (2) at most an $\alpha-$fraction of bad IDs.  We make use of a solution to this problem in Section \ref{sec:genID}. Our algorithms call a solution to \genID only once, at system initialization. A number of algorithms have been found to solve \genID~\cite{andrychowicz2015pow,hou2017randomized,katz2014pseudonymous}.

\medskip
\noindent{\textit {\bf Radio-Resource Testing.}} In a wireless setting with multiple communication channels, Sybil attacks can be mitigated via \emph{radio-resource testing} which relies on the inability of the adversary to listen to many channels simultaneously~\cite{monica:radio,newsome:sybil,gilbert:sybilcast,gilbert:who}.  However, this approach may fail if the adversary can monitor most or all of the channels. Furthermore, radio-resource testing suffers from the same issue discussed above, even in the absence of attack, the system must constantly perform tests. 

\smallskip
\noindent\textit{\bf Social Network Properties.} There are several results that leverage social networks to yield Sybil resistance \cite{yu:survey,lesniewski-laas:whanau,yu:sybilguard,mohaisen:improving,wei:sybildefender,Yu:2010:Sybillimit}. However, social-network information may not be available in many settings. Another idea is to use network measurements to verify the uniqueness of IDs~\cite{bazzi:distinct,sherr:veracity,1550961,liu:mason,Gil-RSS-15,demirbas:rssi}, but these techniques rely on accurate measurements of latency, signal strength, or round-trip times, for example, and this may not always be possible.  Containment strategies are examined in overlays~\cite{danezis:sybil,scheideler:shell}, but the extent to which good participants are protected from the malicious actions of Sybil IDs is limited.

\smallskip
\noindent\textit{\bf Proof of Work and Alternatives.} As a choice for~\POW, computational puzzles provide certain advantages. First, verifying a solution is much easier than solving the puzzle itself. This places the burden of proof on devices who wish to participate in a protocol rather than on a verifier. 

In contrast, bandwidth-oriented schemes, such as~\cite{walfish2010ddos}, require verification that a sufficient number of packets have been received before any service is provided to an ID; this requires effort by the verifier that is proportional to the number of packets. 

A recent alternative to \POW is \defn{proof-of-stake (PoS)} where security relies on the adversary holding a minority stake in an abstract finite resource~\cite{abraham:blockchain}. When making a group decision, PoS weights each participant's vote using its share of a limited resource; for example, the amount of cryptocurrency held by the participant.  A well-known example is ALGORAND~\cite{GiladHMVZ17}, which employs PoS to form a committee.   A hybrid approach using both \POW and PoS has been proposed in the Ethereum system~\cite{ethereum-pos}.  We note that PoS can only be used in systems where the ``stake" of each participant is globally known.  Thus, PoS is typically used only in cryptocurrency applications.


\subsection{Motivation for our \sgoal}


Our algorithms guarantee that the fraction of bad IDs is always less than any fixed constant.  This is in contrast to a popular model where the fraction of bad IDs is {\it assumed} to be bounded below some fixed constant at all times. This latter model features prominently in many areas including Byzantine-resilient replication~\cite{castro1998practical,castro:byzantine,cowling:hq,kotla2007zyzzyva,cachin:secure}, secure peer-to-peer networks~\cite{rodrigues:rosebud,fiat:making,awerbuch:towards,young:towards,guerraoui:highly}, reliable broadcast~\cite{koo,bhandari,koo2,king:sleeping_journal}, and Byzantine consensus~\cite{Kapron:2008:FAB:1347082.1347196,Katz:2009:ECP:1486275.1486420,lamport1982byzantine}.  Thus, our \sgoal can be viewed as a tool enabling these past results to run correctly even under Sybil attack.

      
\section{Commensurate COMputation (\AlgB)}\label{sec:server1} 
In this section, we describe and analyze our first algorithm \AlgB. 

\subsection{GenID}\label{sec:genID}
To initialize our system, we make use of an algorithm created by Andrychowicz and Dziembowski~\cite{andrychowicz2015pow}, which we call \defn{\genID}.  This algorithm creates an initial set of IDs such that no more than an $\alpha$-fraction are bad. \genID  also selects a committee of logarithmic size that has a majority of good IDs.  Finally, \genID has significant, but polynomial, computational cost; thus we use it only once during the lifetime of our system.

\subsection{Constructing Computational Puzzles}\label{subsec:puzzle-construction}

We assume that each good ID can perform $\mu$ hash-function evaluations per round for $\mu>0$. Additionally, we assume that $\mu$ is of some size polynomial in $n_0$ so that $\log \mu = \Theta(\log n_0)$.  It is reasonable to assume large $\mu$ since, in practice, the number of evaluations that can be performed per second is on the order of millions to low billions~\cite{hashcat,mining,non-spec}.

Our technical puzzle generation process follows that of~\cite{andrychowicz2015pow}.   For any integer $h \geq 1$, we define a {\boldmath{$h$}}\defn{-hard puzzle} to consist of finding $C\log \mu$ solutions using a threshold of $h(1-\delta)\mu / (C \log \mu)$, where $\delta>0$ is a small constant and $C$ is a sufficiently large constant depending on $\delta$ and $\mu$.

Let $X$ be a random variable (r.v.) giving the expected number of hash evaluations needed to compute $C \log \mu$ solutions.  Note that $X$ is a negative binomial r.v. and so, the following concentration bound holds (see, for example, Lemma 2.2 in~\cite{awerbuch:towards}) for every $0 < \epsilon \leq 1$:
$$Pr(|X - E(X)| \geq \epsilon E(X)) \leq 2e^{- \epsilon^2 (C \log \mu) / (2(1+\epsilon))}$$

Given the above, one can show that every good ID will solve a $h$-hard puzzle with at most $h\mu$ hash function evaluations, and that the adversary must compute at least $(1-2\delta)h\mu$ hash evaluations to solve every $h$-hard puzzle. This follows from a union bound over $O(n_0^{\gamma})$ join and departure events by IDs, where $\gamma$ is any fixed positive constant, and for $C$ being a sufficiently large constant depending on $\delta, \mu$ and $\gamma$. Note that for small $\delta$, the difference in computational cost is negligible, and that $\mu$ is also unnecessary in comparing costs.  Thus, for ease of exposition, we assume that each {\boldmath{$h$}}-hard puzzle requires computational cost {\boldmath{$h$}} to solve. 

\subsection{How Puzzles Are Used}\label{sec:how-used}  
Although, each puzzle is constructed in the same manner, they are used in two distinct ways by our algorithm. First, when a new ID wishes to join the system, it must provide a solution for a $1$-round puzzle; this is referred to as  \defn{\epuzzle}. Here, the input to the puzzle is {\boldmath{$K_v || s || \tstamp$}}, where {\boldmath{$\tstamp$}} is the timestamp of when the puzzle solution was generated. In order to be verified, the value $\tstamp$ in the solution to an entrance puzzle must be within some small margin of the current time.  In practice, this margin would primarily depend on network latency. 

We note that, in the case of a bad ID, this solution may have been precomputed by the adversary by using a future timestamp. This is not a problem since the purpose of this puzzle is only to force the adversary to incur a computational cost at some point, and to deter the adversary from reusing puzzle solutions. Importantly, the \epuzzle is not used to preserve guarantees on the fraction of good IDs in the system.

The second way in which puzzles are used is to limit the fraction of bad IDs in the system; this is referred to as a \defn{\ppuzzle}.  An announcement is periodically made that \emph{all} IDs {\it already in the system} should solve a 1-round puzzle. When this occurs, a \defn{random string} {\boldmath{$r$}} of $\Theta(\log n_0^{\gamma})$ bits is generated and included as part of the announcement. The value $r$ must also be appended to the inputs for all requested solutions in this round; that is, the input is {\boldmath{$K_v || s || r$}}.  These random bits ensure that the adversary cannot engage in a pre-computation attack~\textemdash~where it solves puzzles and stores the solutions far in advance of launching an attack~\textemdash~by keeping the puzzles unpredictable. For ease of exposition, we omit further discussion of this issue and consider the use of these random bits as implicit whenever a \ppuzzle is issued. 

While the same $r$ is used in the puzzle construction for all IDs, we emphasize that a {\it different} puzzle is assigned to each ID since the public key used in the construction is unique. Again, this is only of importance to the second way in which puzzles are used. Using the public key in the puzzle construction also prevents puzzle solutions from being stolen. That is, ID $K_v$ cannot lay claim to a solution found by ID $K_w$ since the solution is tied to the public key $K_w$. 

Can a message $m_v$ from ID $K_v$ be spoofed? This is prevented in the following manner.  ID $K_v$ signs $m_v$ with its private key to get signature $\texttt{sign}_v$, and then sends $(m_v || \texttt{sign}_v || K_v )$ via \Diffuse. Any other ID can use $K_v$ to check that the message was signed by the ID $K_v$ and thus be assured that ID $K_v$ is the sender.


\begin{figure}[h!]
\begin{tcolorbox}[standard jigsaw, opacityback=0]

\begin{minipage}[h]{0.94\linewidth}
\noindent{\bf \large Algorithm \AlgB}\smallskip

\vspace{5pt}\noindent{\bf Key Variables:}\\
\noindent{}{$i$ : iteration number}\\
\hspace{-7pt}\noindent{}{$n_i^a$: number of IDs joining since beginning of iteration $i$.}\\
\hspace{-7pt}\noindent{}{$n_i^d$: number of IDs departing since beginning of iteration $i$.}\\
\hspace{-7pt}\noindent{}{$\Old$: set of IDs at the end of iteration $i$.}\\
\hspace{-7pt}\noindent{}{$\curIDs$: current set of IDs}.\\
\hspace{-7pt}\noindent{}{$r$: current random seed}.\medskip

\noindent{\bf Initialization:}\\
\noindent $\setIDs_0 \leftarrow$ set of IDs returned by initial call to \genID. The initial committee is also selected by \genID.\\ 
\noindent $i \leftarrow 1$.\\

\vspace{-5pt}
\hspace{0pt}\noindent The committee maintains all variables above and makes all decisions using Byzantine Consensus:
\medskip
\begin{enumerate}[leftmargin=15pt]

	\item Each joining ID, $v$, generates a public key; solves an \epuzzle using that key and the current timestamp $\tstamp$ and broadcasts the solution via \Diffuse. Upon verifying the solution, the committee adds $v$ to $\curIDs$.\medskip

\item If $n_i^a + n_i^d  \geq |\mathcal{S}_{i-1}|/3$, then the current committee does the following:\label{s:testStep}\smallskip
		
	\textbf{Perform Purge:}
	\begin{itemize}[leftmargin=17pt]
		\item[(a)] Generates a random string $r$ and broadcast via \Diffuse.
		\item[(b)] Sets $\Old$ and $\curIDs$ to the set of IDs that return valid solutions.
		\item[(c)] Selects a new committee of size $\Theta(\log n_0)$ from $\Old$ and sends out this information via \Diffuse.\smallskip
	\end{itemize}

\item $i \leftarrow i + 1$\smallskip
		
\end{enumerate}
\end{minipage}

\end{tcolorbox}
\caption{Pseudocode for \textsc{Commensurate COMputation (\AlgB)}.}
\label{alg:ccom}
\end{figure}


\subsection{Overview of \AlgB}\label{sec:overview}

We begin by describing  \emph{{\underline{C}}ommensurate} \emph{{\underline{Com}}put- ation} (\AlgB). The execution of \AlgB~is conceptually broken into sets of consecutive rounds called \defn{iterations}. Each iteration consists of  Steps 1 and 2, and each iteration ends with the execution of Step 2(d). 

Notation is introduced as needed.  Throughout, we let $\log$ denote the natural logarithm.

\subsubsection{Updating and Maintaining System Membership}\label{sec:sys-mem}

A key component of our system is a subset of IDs called a \qcomm for which the \cgoal must hold; recall Section~\ref{sec:our-problems}.  The \qcomm tracks the current membership in the system denoted by the set {\boldmath{$\curIDs$}}. Updates to this set are performed in the following manner. Each ID that wishes to join the system must solve an \epuzzle. Each good ID in the \qcomm will check the solution to the \epuzzle and, upon verification, the joining ID is allowed into the system. 

Verification of a puzzle solution requires checking that (1) all $C\log \mu$ inputs to $h$ submitted generate an output that is at most $(C\log \mu)/((1-\delta)\mu)$; and (2) each of these inputs contains the string $r$ and also the individual public key of the ID. Upon verification, the good IDs in the committee solve Byzantine Consensus in order to agree on the contents of $\curIDs$.  Those IDs that fail to submit valid puzzle solutions are denied entrance into the system.

Similarly,  $\curIDs$ is also updated and agreed upon when a good ID informs the server that it is departing. Of course, bad IDs may not provide such a notification and, therefore, $\curIDs$ is not necessarily accurate at all times. Finally, if we wish to achieve only the Committee Goal, then the last portion of Line 1 is omitted: there is no need for the committee to verify solutions or maintain $\curIDs$.

At the beginning of the system, the \qcomm knows the existing membership denoted by $S_0$; assume {\boldmath{$|S_0|=n_0$}} initially.  At some point, $n_i^a + n_i^d \geq |S_0|/3$, where {\boldmath{$n_i^a$}} and {\boldmath{$n_i^d$}} denote the number of join events and departure events over iteration $i$; note that if the same ID joins (and departs) more than once, each such event increases $n_i^a$ and  $n_i^d$. 

At this point, Step 2 executes, whereby all IDs are issued a  purge puzzle and each ID must respond with a valid solution within $1$ round. The issuing of these \ppuzzle{}s is performed by the \qcomm via the diffusing of $r$ to all IDs; this random string is created via solving Byzantine Consensus (for example, the protocol in~\cite{Katz:2009:ECP:1486275.1486420}) in order to overcome malicious behavior by bad IDs that are committee members.\footnote{Agreeing on each random bit of $r$ will suffice. Alternatively, a secure multiparty protocol for generating random values can be used; for example, the result in~\cite{srinathan_pandu_rangan:efficient}.} 

Recall from Section~\ref{sec:model-main} that a round is of sufficient length that the computation time to solve a $1$-round puzzle dominates the round trip communication time between the client and server (we remove this assumption later in Section~\ref{sec:bounded-latency}). Thus, the good IDs can receive solutions to purge puzzles in a round, and then update and agree on $\curIDs$ at the end of the iteration.

\begin{figure*}[t!]
\centering
\includegraphics[width = \textwidth]{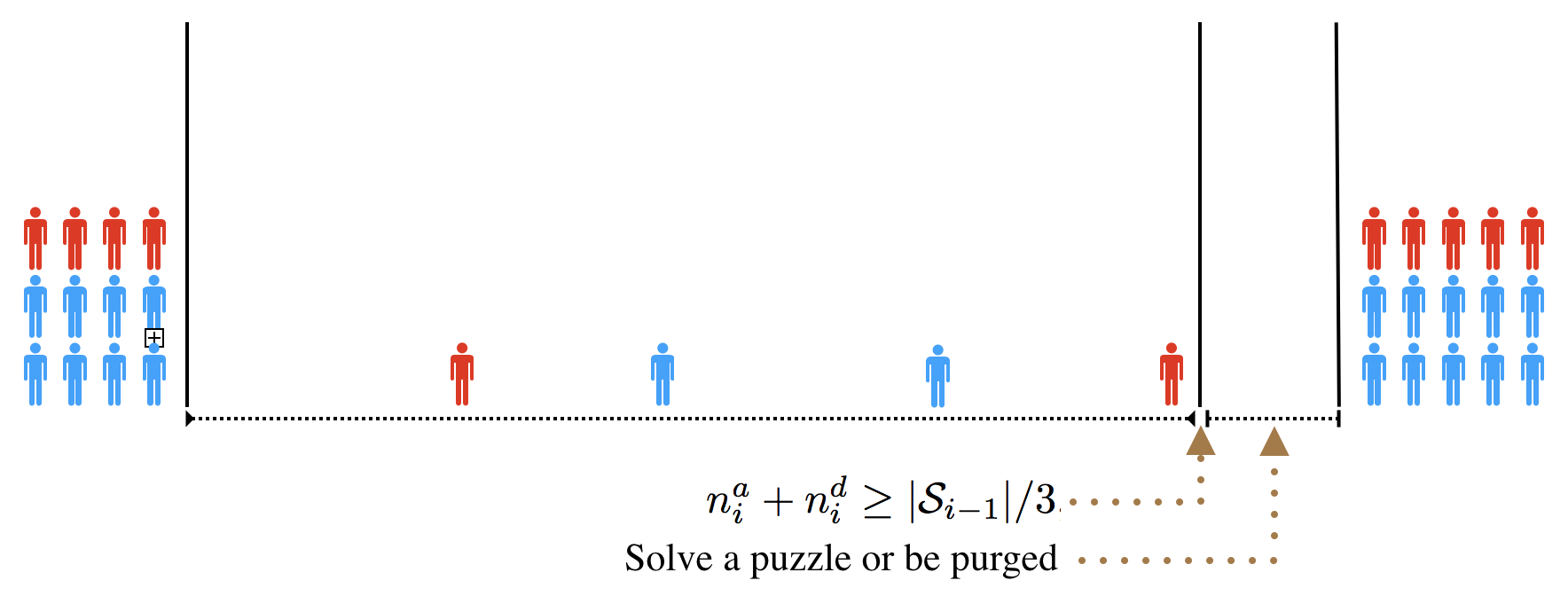}
\vspace{-0.5cm}\caption{An overview of our approach. The system starts with $x=12$ IDs of which $4$ are bad and $8$ are good. When the $(x/3)^{\mbox{\tiny th}} = 4^{\mbox{\tiny th}}$ ID joins, a purge is immediately executed and all IDs must solve a \ppuzzle or be ejected from the system.}
  \label{f:overview}  
  \vspace{-20pt}
\end{figure*}

Figure~\ref{f:overview} illustrates this purging process;  good IDs are blue and bad IDs are red.   Let $x$ be the number of IDs after the most recent purge; at most $x/3$ of these IDs are bad, since $\alpha=1/3$. According to Step 2 in \AlgB, the next purge occurs immediately when the number of new IDs entering would become greater than or equal to $x/3$.  At the point when a purge is triggered, the number of new IDs that have been allowed to enter is less than $x/3$.  Thus, the total number of bad IDs is less than $\frac{2}{3}x$, and the total number of good IDs is at least $\frac{2}{3}x$. 

How is $r$ sent such that each good ID trusts it came from the committee? Recall that a  participant $v$ that joins the system generates a public key, $K_v$, and the corresponding private key, $k_v$; the public key is used as its ID. If ID $K_v$ is a committee member, it will sign the random string $r$ using its private key $k_v$ to obtain $\texttt{sign}_v$. Then, $(r || \texttt{sign}_v || K_v)$ is sent using \Diffuse.  Any good ID can verify $\texttt{sign}_v$ via $K_v$ to ensure it returns $r$. This implicitly occurs in Step 2 of Figure~\ref{alg:ccom}, but we omit the details in the pseudocode for simplicity.

\subsubsection{Updating and Maintaining Committee Membership}\label{sec:com-mem}

Recall from Section \ref{sec:model-main} that good IDs always inform all other IDs of their departure. Thus, a good ID $K_v$ that is a committee member will inform other committee members of its departure. Since bad IDs may not inform of their departure, our estimate of the number of IDs that join and leave will not be accurate but this will not jeopardize the committee goal.

The committee maintenance operations occur in Step 2 of \AlgB.  In Step 2(a), the committee generates and diffuses a random string $r$; recall the details discussed in Section~\ref{sec:how-used}.

In Step 2(b), each ID must respond with a valid solution within $1$ round, or else be removed from the system. The committee removes unresponsive or late IDs from its whitelist, which is maintained via Byzantine consensus amongst the committee members. Over the lifetime of the system, the \qcomm~should always satisfy the \cgoal. As joins and departures occur, this invariant may become endangered. To address this issue, \qcomm{}s are disbanded and rebuilt over time. In particular,  a new \qcomm is completed and ready to be used by the start of each successive iteration, and the currently-used \qcomm is disbanded. The construction of a new \qcomm occurs in Step 2(c).

In Step 2(c), the current committee then selects $\Theta(\log n_0)$ IDs uniformly at random from  $\setIDs_{i}$. This selection is done via a committee election protocol, for example~\cite{Katz:2009:ECP:1486275.1486420, Kapron:2008:FAB:1347082.1347196}. The committee uses \Diffuse to inform $\setIDs_{i}$ that the selected IDs are the new committee for iteration $i+1$.  All messages from committee members are verified via public key digital signatures.

Given that the adversary has a minority of the computational power, we expect the newly-formed committee to have a good majority.  Byzantine consensus allows for agreement on the members of the new committee, and this is communicated to all good IDs in the network via \Diffuse.  These two properties satisfy the other two criteria of the \cgoal; this is argued formally in Section~\ref{sec:committee-goodness}.

Over the round during which this construction is taking place, the existing committee still satisfies the \cgoal; this is proved in Lemma~\ref{lem:maj_comm}. Once the new committee members are elected, they are known to all good IDs in the system. Those committee members that were present prior to this construction are no longer recognized by the good IDs. Finally, any messages that are received after this single round are considered late and discarded.


\setlength\extrarowheight{2pt}
\newcolumntype{P}[1]{>{\centering\arraybackslash}p{#1}}
\begin{table}[t!]
  \begin{center}
    \vspace{-5pt}
    \label{tab:notation1}
    {\def\arraystretch{1.1}
    {
\begin{tabular}{ |P{1.5cm}|p{10.5cm}|  }
   \hline
 \multicolumn{2}{|c|}{{\bf Used in Model and Algorithm}}\\
     \hline
      \textbf{Symbol} &   \textbf{Description}\\
    \hline
      $\alpha$ & Fraction of total computational power that the adversary controls. \\
            \hline
       $n_0$ & Lower bound on number of good IDs in the system at any time.\\ 
      \hline
      $\curIDs$ & Set of all IDs in the system at current time.\\
            \hline
            $\setIDs_i$ & Set of all IDs in the system at the end of iteration $i$. Equivalently, the set at the beginning of iteration $i+1$. \\ 
            \hline
                  $\ell_i$ & Length of iteration $i$.\\
      \hline
    \end{tabular}
    }} 
    \vspace{10pt}
    \caption{A summary of notation for \AlgB}
  \end{center}
\end{table}


\subsection{Analysis of \AlgB}

To simplify our presentation, our claims are proved to hold with probability at least $1-\tilde{O}(1/n_0^{\gamma + 2})$, where $\tilde{O}$ hides a poly($\log n_0$) factor.  Of course, we wish the claims of Theorem~\ref{thm:main1} to hold with probability at least $1- 1/n_0^{\gamma+1}$ such that a union bound over $n_0^{\gamma}$ joins and departures yields a w.h.p. guarantee.  By providing this ``slack'' of an $\tilde{\Omega}(1/n_0)$-factor in each of the guarantees of this section, we demonstrate this is feasible while avoiding an analysis cluttered with specific settings for the constants used in our arguments.

For any iteration $i$, let {\boldmath{$G_i$}} and {\boldmath{$B_i$}} denote the set of good IDs and bad IDs, respectively at the end of iteration $i$.

\subsubsection{Analysis of the \cgoal}\label{sec:committee-goodness}

In this section, we prove that \whp~\AlgB~preserves the~\cgoal over a polynomial number of join and leave events.

\begin{lemma}\label{lem:maj_comm}
With probability $1 - O(1/n_0)$, over a polynomial number of join and departure events,  there is always an honest majority in the committee. \end{lemma}

\begin{proof}
For iteration  $i = 0$, the committee invariant holds by the correctness of \genID
(recall Section~\ref{sec:genID}).
	
Fix an iteration $i > 0$. Recall that a new committee is elected by the existing committee by selecting $c\log |S_i|$ IDs independently and uniformly at random from the set $S_i$, for a sufficiently large constant $c>0$ which we define concretely later on in this proof. Let $X_G$ be a random variable which denotes the number of good IDs elected to the new committee in iteration $i$. Then:
	\begin{eqnarray}\label{eq:expgc}
		E[X_G] & = & \frac{|G_i|}{|S_i|}c\log{|S_i|} \text{ } = \text{ } (1-\alpha)c\log|S_i| 
	\end{eqnarray}
	 where the last inequality follows from the fact that the fraction of computational power with the adversary is at most $\alpha$. Next, we bound the number of good IDs in the committee using Chernoff Bounds~\cite{mitzenmacher2017probability}:
	 \begin{eqnarray*}
	 	Pr \left( X_G < (1-\delta)(1-\alpha)c\log|S_i| \right) \leq  \exp{\left\{ -\frac{\delta^2 (1-\alpha) c\log|S_i|}{2} \right\}} =  O\left({n_0^{-(\gamma + 1)}}\right)
	 \end{eqnarray*}
 where the first step holds for any constant $0 < \delta < 1$, the second step follows from Equation \ref{eq:expgc}, and the last step holds for all $c \geq \frac{28}{13}\frac{(\gamma + 1)}{\delta^2}$.  For any fixed, constant $\epsilon_0>0$, there exists a $\delta$, such that the number of good IDs in the committee is at least $(1- \alpha - \epsilon_0) c\log{|S_i|}$ with probability  $1-O(n_0^{-(\gamma + 1)})$.
 	 	
Next we bound departures of good IDs from the committee over an iteration. Let $Y_g$ be a random variable which denotes the number of good IDs that depart from the committee during the iteration, i.e. up to the point when the number of departures of IDs from the system is $|S_{i-1}|/3$.  Departing good IDs are assumed to be selected uniformly at random from the population of good IDs (Section~\ref{sec:join}). Thus, we obtain:
		\begin{eqnarray}\label{eq:dgoodc}
			E[Y_g] & \leq & \frac{|S_i|}{3}\left(\frac{c\log|S_i|}{|S_i|}\right) \text{ } = \text{ } \frac{c}{3} \log|S_i|
		\end{eqnarray}  
		Next, we  upper bound  the number of departures of good IDs from the committee using Chernoff Bounds~\cite{mitzenmacher2017probability}:
		\begin{eqnarray*}
			Pr \left( Y_g > {(1+\delta')} \frac{c \log|S_i|}{3} \right) &\leq&  \exp\left\{ -\frac{\delta'^2c}{9}\log|S_i| \right\} = O\left( {n_0^{-(\gamma+1)}} \right)
		\end{eqnarray*}
		where the first step holds for any constant $0< \delta' < 1$ and the last step holds for all $c \geq \frac{9(\gamma + 1)}{\delta'^2}$. 

Thus, for any fixed $\epsilon_1>0$, the minimum number of good IDs in the committee at any point during the iteration is greater than $(1 - \alpha - 1/3 - \epsilon_0 - \epsilon_1) c\log{|S_i|}$. Good IDs can also depart while the committee is running the single-round purge at the end of iteration.  But, by assumption, at most an $\epsilon_d$-fraction of good IDs depart per round, where $\epsilon_d$ is sufficiently small. Therefore,  the number of good IDs exceeds $(1 - \alpha - 1/3 - \epsilon_0 - \epsilon_1 - \epsilon_d) c\log{|S_i|}$. Choosing $\alpha < 1/6$ ensures that we can choose $\epsilon_0$,  $\epsilon_1$, and $\epsilon_d$ such that the fraction of good IDs is always strictly greater than $1/2$.
Finally, by a union bound over $n_0^\gamma$ iterations, the committee invariant is maintained over $n_0^\gamma$ iterations with probability $1 - O(1/n_0)$, which implies the claim.
\end{proof}


\subsubsection{Analysis of the \sgoal}\label{sec:full-dynamism}

\begin{lemma}\label{lem:bound_b}
For all iterations $i \geq 0$,  $|B_i| \leq \alpha |S_{{\tiny i}}|$.
\end{lemma}
\begin{proof}
This is true by the properties of \genID for $i=0$.  Then after Step 2 that ends iteration $i$ for $i>0$, we have $|B_i| \leq \alpha |S_i|$.
\end{proof}

Let {\boldmath{$g_i^a$}} and {\boldmath{$b_i^a$}} denote the good and bad IDs that join over iteration $i$, respectively; and let {\boldmath{$n_i^a$}} $= g_i^a + b_i^a$. Similarly, let  {\boldmath{$g_i^d$}} and {\boldmath{$b_i^d$}} denote the good and bad IDs that depart over iteration $i$, respectively;  and let {\boldmath{$n_i^d$}} $=  g_i^d + b_i^d$. Note that the \qcomm will have an accurate value for all of these variables except for possibly $b_i^d$ and, consequently,  $n_i^d$, since bad IDs may not give notification when they depart. Therefore, for $b_i^d$ and  $n_i^d$, the \qcomm may hold values which are underestimates of the true values. However, we highlight that, in our argument below, this is not a problem for maintaining the Population Invariant.

\begin{lemma}\label{lem:badlesshalf-full}
The fraction of bad IDs is always less than $3/8$. 
\end{lemma}
\begin{proof}
Fix some iteration $i>0$. Step 1 in iteration $i$ ends when $n_i^a + n_i^d \geq |\setIDs_{i-1}|/3$. Therefore, we have $b_i^a + g_i^d \leq |\setIDs_{i-1}|/3$.  We are interested in the maximum value of the ratio of bad IDs to total IDs at any point during the iteration.  Thus, we pessimistically assume all additions of bad IDs and removals of good IDs come first in the iteration.  We are then interested in the maximum value of the ratio:\vspace{-10pt}

$$
\frac{|B_{i-1}| + b_i^a}{|\setIDs_{i-1}| + b_i^a - g_i^d}. 
$$

By Lemma~\ref{lem:bound_b}, $|B_{i-1}| \leq \alpha |\setIDs_{i-1}|$.
 Thus, we want to find the maximum of $\frac{\alpha |\setIDs_{i-1}| + b_i^a}{|\setIDs_{i-1}| + b_i^a - g_i^d}$, subject to the constraint that $b_i^a + g_i^d \leq |\setIDs_{i-1}|/3$. This ratio is maximized when the constraint achieves equality, that is when $g_i^d = (|\setIDs_{i-1}|/3) - b_i^a$.  Plugging this back into the ratio, we get \vspace{-10pt}

\begin{eqnarray*}
\frac{\alpha |\setIDs_{i-1}| + b_i^a}{|\setIDs_{i-1}| + b_i^a - g_i^d}
 &\leq & \frac{\alpha|\setIDs_{i-1}| + b_i^a}{2|\setIDs_{i-1}|/3 + 2b_i^a}\\
& \leq & \frac{\alpha|\setIDs_{i-1}| + |\setIDs_{i-1}|/3}{2|\setIDs_{i-1}|/3 + 2|\setIDs_{i-1}|/3} \\
& < & 3/8
\end{eqnarray*}

The second line of the inequalities holds since $b_i^a \leq |\setIDs_{i-1}|$, and the last line holds since $\alpha \leq 1/6$.

Finally, we note that this argument is valid even though $\curIDs$ may not account for bad IDs that have departed {\it without} notifying the \qcomm (recall this is possible as stated in Section~\ref{sec:model-main}). Intuitively, this is not a problem since such departures can only lower the fraction of bad IDs in the system; formally, the critical equation in the above argument is $b_i^a + g_i^d \leq |S_{{\tiny i-1}}|/3$, and this does not depend on $b_i^d$.
\end{proof} 


\subsubsection{Cost Analysis and Proof of Theorem~\ref{thm:main1}}

We now examine the cost of running \AlgB. Let $T_i$ be the total computational cost incurred by the adversary in iteration $i$. 

\begin{lemma}\label{lemma:committee-cost}  
	For any iteration $i \geq 1$, the total computational cost to the good IDs  is $O(T_i + g_i^a+ g_i^d)$.
\end{lemma}
\begin{proof} 
Fix an iteration $i\geq 1$. Let $\mathcal{A}_i$ denote the total computational cost to the good IDs  in this iteration. Note that 
Note that the total computational cost to the algorithm consists of the total purge cost paid by good IDs in the system at the end of this iteration {plus} the total entrance cost paid by the new good IDs that join the system in this iteration. Thus, we have:

\begin{eqnarray*}
	\mathcal{A}_i & = &  |G_i| + g_i^a\\
	 & \leq & \frac{4}{3} |S_{{\tiny i-1}}| + g_i^a\\
	 & \leq & 8(g_i^a + g_i^d + 2b_i^a)  + g_i^a\\
	 & = & 9g_i^a + 8g_i^d + 16T_i
\end{eqnarray*}

In the above, the second inequality holds since $|G_i| \leq |S_i|$ and $|S_i| \leq \frac{4}{3}|S_{i-1}|$.

The third inequality holds because $b_i^d \leq \alpha |S_{i-1}| + b_i^a$, since $|B_{i-1}| \leq \alpha |S_{i-1}|$; and also $g_i^a + g_i^d + b_i^a + b_i^d \geq |S_{i-1}|/3$.  Thus we have that $g_i^a + g_i^d + 2b_i^a \geq (1/3 - \alpha) |S_{i-1}| > (1/6) |S_{i-1}|$.

Finally, the last inequality holds since $b_i^a \leq T_i$. 
\end{proof}


\noindent Pulling the pieces together from the previous sections, we can now prove Theorem~\ref{thm:main1}.

\begin{proof}[Proof of Theorem~\ref{thm:main1}]
{We first prove that the population and committee goals hold over the lifetime of the system, followed by the analysis of cost to the algorithm.}\smallskip

\noindent\textbf{Population and Committee goals.} By Lemma~\ref{sec:committee-goodness}, with probability at least $1-O(1/n_0^{\gamma+1})$, the \qcomm has a majority of good IDs, has size $\Theta(\log n_0)$, and its members are communicated to all good IDs via \Diffuse; therefore, the \cgoal is met. Given that the \cgoal holds, by Lemma~\ref{lem:badlesshalf-full} the fraction of bad IDs is less than $1/2$ over the $O(n_0^\gamma)$ join and leave events; therefore, the \sgoal is met.
\smallskip  

\noindent\textbf{Spending Rate.} Let $x$ be the number of iterations over the lifetime of the system. Then, from Lemma \ref{lemma:committee-cost}, we can obtain the total computational cost to the algorithm to be at most:
	\begin{align*}
		\sum_{i = 1}^x O\left( T_i + g_i^a + g_i^d  \right) \leq O\left( \sum_{i = 1}^x T_i + \sum_{i = 1}^x g_i^a + \sum_{i = 1}^x g_i^d\right)
	\end{align*}
	Note that the number of departing IDs from a system is at most equal to the number of IDs that join the system over the lifetime of the system. Using this, we obtain the total computational cost to the algorithm as:
	\begin{align*}
		 O\left( \sum_{i = 1}^x T_i + \sum_{i = 1}^x g_i^a\right) 
	\end{align*}
	Let {\boldmath{$\ell_i$ }}denote the length of the $i^{\mbox{\tiny th}}$ iteration. Then, we can obtain the average cost to the algorithm over the lifetime of the system as:
	\begin{align*}
		 &\frac{O\left( \sum_{i = 1}^x T_i + \sum_{i = 1}^x g_i^a\right)}{\sum_{i=1}^x \ell_i} \\
		 &\leq O\left(\frac{\sum_{i=1}^x T_i}{\sum_{i=1}^x \ell_i} + \frac{\sum_{i=1}^x g_i^a}{\sum_{i=1}^x \ell_i}\right) \\
		 &= O\left( T + \joinRate \right)
	\end{align*}
	\noindent The last step follows since, over all iterations over the lifetime of the system, we have  $T$ $= \left(\sum_{i = 1}^x T_i \right)/\left({\sum_{i=1}^x {\ell}_i}\right)$ and $\joinRate =  \left(\sum_{i=1}^x g_i^a \right)/\left(\sum_{i=1}^x \ell_i\right)$.
	\end{proof}


\section{Generalizations of \AlgB}\label{sec:enhancements}
In this section, we present two enhancement to \AlgB. The first allows for our results to hold when there is {a known} bounded communication latency.  The second is a general rule for when to execute Step 2 in order to maintain a desired strict upper bound (not necessarily $1/2$) on the fraction of bad IDs in the system; that is, this a more flexible Population Invariant.

\subsection{Handling Bounded Communication Latency}\label{sec:bounded-latency}

Previously, all messages were assumed to be delivered with negligible delay when compared with the time to solve computational puzzles. In this section, we relax this assumption to the following:   \textit{the adversary may choose to arbitrarily delay messages, subject to the constraint that they are delivered within a bounded number of rounds $\Delta>0$. } Such a model of communication delay is often referred to as \defn{partial synchrony}  or {\boldmath{$\Delta$}}\defn{-bounded delay} in the literature \cite{dwork1988consensus,pass2017analysis}. We address this issue in the following analysis for each of our algorithms.\medskip

\noindent\textbf{In \textsc{Distributed} \AlgB.} In \AlgB, messages are exchanged with good IDs every time the committee is constructed. Thus, we increase the wait time for receiving the solutions to the \ppuzzle by an additive term of $2\Delta$.  To simplify the analysis of a $\Delta$-bounded delay in the network where the adversary holds  an $\alpha$-fraction of the total computational power, we can instead consider a $0$-bounded delay network, but where the adversary holds an increased fraction of the total computational power; this is addressed by Lemma \ref{boundedfrac}. 

\begin{lemma}{\label{boundedfrac}}
A $\Delta$-bounded network latency with an adversary that holds an $\alpha$-fraction of total computational power of the network is equivalent to a $0$-bounded network latency with an adversary that holds a $\frac{2\alpha\Delta+\alpha}{2\alpha\Delta + 1}$-fraction of the total computational power of the network.
\end{lemma}
\begin{proof}
Recall from Section~\ref{sec:model-main}, the fraction of computational power with an adversary is $\alpha$-times the total computational power of the network. Given the $\Delta$-bounded delay in the network, the adversary has $2\Delta+1$ rounds to perform evaluations. Fix a time step during an iteration.  Suppose $P$ is the total computational power of the network at that time step. Then, the number of evaluations of the hash function an adversary can make is $P\alpha(2\Delta + 1)\mu$.

Note that the computational power of the good IDs remains the same, which is $P(1-\alpha)\mu$ since they obey the protocol, and we assume that the bounded delay only benefits the adversary. Thus, the $\Delta$-bounded delay model is equivalent to the $0$-bounded delay model with an adversary that holds a {$\frac{P\alpha(2\Delta + 1)\mu}{P\alpha(2\Delta + 1)\mu + P(1-\alpha)\mu} = \frac{2\alpha\Delta+\alpha}{2\alpha\Delta + 1}$}-fraction of the computational power.
\end{proof}

\begin{lemma}{\label{bounded-dist}}
Assume a $\Delta$-bounded delay and an adversary that holds an $\alpha$ - fraction of computational power, for $\alpha \leq 1/(10\Delta + 6)$. With probability at least $1-O(1/n_0^{\gamma+1})$, the fraction of good IDs in the committee exceeds $1/2$ for the duration of an iteration.
\end{lemma}
\begin{proof}
From Lemma~\ref{boundedfrac}, the $\Delta$-bounded delay is equivalent to a $0$-bounded delay with  an adversary that holds a $\frac{2\alpha\Delta+\alpha}{2\alpha\Delta + 1}$-fraction of computational power of the network. By Lemma~\ref{lem:maj_comm}, for a $0$-bounded delay model  with an adversary that holds an $\alpha'$-fraction of the computational power for $\alpha'\leq 1/6$, the fraction of good IDs exceeds $1/2$. Substituting $\alpha' = \frac{2\alpha\Delta+\alpha}{2\alpha\Delta + 1}$, the $\Delta$-bounded delay model with an $\alpha$-fraction of computational power has at least $1/2$-fraction of good IDs, for $\alpha \leq 1/(10\Delta+6)$.
\end{proof}


\subsection{Reducing the Fraction of Bad IDs in the Population Invariant} \label{sec:popgoalmod}

We show how to modify our algorithm in order to maintain the fraction of bad IDs in the system to be less than $3\alpha$ at all times for any given $\alpha < 1/6$ (i.e., fraction of computational power with the adversary is some value strictly less than 1/6 of the total computation power in the system). We also show the analogous result for the Committee Invariant.

In order to maintain the fraction of bad IDs to be always some fraction less than $3\alpha$, we propose the following modification to Step 2 of Figure~\ref{alg:ccom}. Instead of $n_i^a + n_i^d \geq |S_{\text{\tiny i-1}}|/3$, we  have the condition  $n_i^a + n_i^d \geq \frac{2\alpha}{1 + 3 \alpha}|S_{\text{\tiny i-1}}|$. \smallskip

We prove the more general Population Invariant in the following lemma:

\begin{lemma}\label{lem:modifiedpop}
	For any $\alpha < 1/6$, the fraction of bad IDs in the system is always less than $3\alpha$.
\end{lemma} 
\begin{proof}
	Fix some iteration $i \geq 1$. Then, similar to Lemma \ref{lem:badlesshalf-full}, we can calculate the maximum fraction of bad IDs in iteration $i$ as:
	$$\frac{|B_{\text{\tiny i-1}}| + b_i^a}{|S_{\text{\tiny i-1}}| + b_i^a - g_i^d}$$
Note that $|B_\text{\tiny{i-1}}| < \alpha |S_{\text{\tiny i-1}}|$ at the end of iteration $i-1$. Also, in order to handle the variant Population Invariant, we modify our condition to end an iteration to $n_i^a + n_i^d \geq \frac{2\alpha}{1 + 3\alpha}|S_{\text{\tiny i-1}}|$ for a given $\beta < 3 \alpha$. Thus, the above fraction is maximized when $b_i^a + g_i^d = \frac{2\alpha}{1 + 3\alpha}|S_{\text{\tiny i-1}}|$.  It follows that:
	\begin{eqnarray*}
	\frac{|B_{\text{\tiny i-1}}| + b_i^a}{|S_{\text{\tiny i-1}}| + b_i^a - g_i^d} & < &\frac{\alpha|S_{\text{\tiny i-1}}| + b_i^a}{|S_{\text{\tiny i-1}}| + b_i^a - g_i^d}\\
	& = & \frac{\alpha|S_{\text{\tiny i-1}}| + b_i^a}{   \left(\frac{2\alpha}{1+3\alpha} + \frac{1+\alpha}{1+3 \alpha}\right)   |S_{\text{\tiny i-1}}| + b_i^a - g_i^d}\\
	&=&\frac{\alpha|S_{\text{\tiny i-1}}| + b_i^a}{\frac{1+\alpha}{1+3\alpha}|S_{\text{\tiny i-1}}| + 2b_i^a} \\
	&\leq & \frac{\alpha|S_{\text{\tiny i-1}}| + b_i^a}{\frac{1+\alpha}{1+3\alpha}|S_{\text{\tiny i-1}}|}\\
	&=& \frac{1+3\alpha}{1+\alpha}\left(\alpha + \frac{b_i^a}{|S_{\text{\tiny i-1}}|}\right)\\
	&<&3\alpha
	\end{eqnarray*}
	
	\noindent Where the final  inequality holds since ${b_i^a} < \frac{2\alpha}{1 + 3\alpha}|S_{\text{\tiny i-1}}|$. 
\end{proof}


\section{Geometric Mean COMputation (\algGM)}\label{sec:gmcom} 

We now describe our algorithm, \algGM.  Recall that  \algGM makes additional assumptions about the join and departures of good IDs, as detailed in Section~\ref{sec:model-main}.  

\begin{figure}[t!]
\begin{tcolorbox}[standard jigsaw, opacityback=0]
\begin{minipage}[h]{0.95\linewidth}
\noindent{\bf\large Algorithm \algGM}\smallskip

\vspace{5pt}
\noindent{}\hspace{-3pt} {\bf Key Variables}\\
\noindent{} $i$: \hspace{0pt} iteration number\\
\hspace{-7pt}\noindent{}$n_i^a$: number of IDs joining since beginning of iteration $i$\\
\hspace{-7pt}\noindent{}$n_i^d$: number of IDs departing since beginning of iteration $i$\\
\noindent{}  $\iterIDs_i$: set of IDs at end of iteration  $i$ \\
\noindent{}$\curIDs$: current set of IDs in system\\

\noindent{\bf Initialization:}\\
\noindent $i \leftarrow 1$.\\
\noindent $\setIDs_0 \leftarrow$ set of IDs returned by initial call to \genID. The initial committee is also\newline \mbox{\hspace{26pt}selected by \genID}\\ 
\noindent $\JoinEst_{0}  \leftarrow$ obtained in initialization phase  \\
\noindent $\estSet \leftarrow \setIDs_0$ 
\medskip

\noindent{\bf Execution:}\smallskip

\noindent The committee maintains all variables above and makes all decisions using Byzantine Consensus. The committee also maintains: \smallskip
\begin{enumerate}[label=(\roman*), leftmargin=20pt]
\item $\estSet \leftarrow$ \mbox{most recent system membership ensuring $|\curIDs - \estSet | \geq  \frac{3}{5}|\curIDs|$.}
\item $\estDur\leftarrow$ \mbox{length of time between last two changes of variable $\estSet$.}\\ 
\end{enumerate}

\noindent For each iteration $i$, do:\smallskip

\begin{enumerate}[leftmargin=14pt] 

	\item[1.] Each joining ID solves and diffuses the solution to an entrance puzzle of difficulty equal to the number of IDs that have joined in the last $1/\JoinEst_{i}$ seconds of the current iteration, including the newly joining ID.\smallskip
	
	\item[2.] When  $n_i^a + n_i^d \geq (1/11)|\setIDs_{i-1}|$, do: 
			\textbf{Perform Purge:}
			\begin{itemize}[leftmargin=20pt]
			\item[(a)] The committee generates and diffuses a random string $r$ to be used in puzzles for this purge and entrance for the next iteration. 
			\item[(b)] $\iterIDs_i$ $\leftarrow$  set of IDs returning difficulty $1$ puzzle solutions within $1$ round.
			\item[(c)] The committee selects a new committee of size $\Theta(\log n_0)$ from $\iterIDs_{i}$ and sends out this information via \Diffuse.\smallskip
			\end{itemize}
			
\textbf{End Iteration:}
\begin{itemize}
\item[(d)] $i \leftarrow i+1$
\item[(e)]  $\JoinEst_i \leftarrow |\curIDs|/\estDur$ 
\end{itemize}
						
\end{enumerate}
\end{minipage}
\end{tcolorbox}
\caption{Pseudocode for \textsc{GeometricMean COMputation (\algGM)}.}
\label{alg:gmcom}
\end{figure}

In this section, we first discuss a key technical challenge of \algGM: estimating the good join rate (Section~\ref{subsec:estimate-good-join-rate}).  Next, we give an overview of \algGM (Section~\ref{sec:overview}), see Figure~\ref{alg:gmcom} for pseudocode.  Then we discuss intuition for our new entrance costs (Section~\ref{sec:entrance-cost}). We collect notation specific to this section in Table~\ref{tab:notation2} for easy reference.

\medskip


\subsection{\boldmath{$\iJRate_i$} and {\boldmath{$\JoinEst_{i}$}}}\label{subsec:estimate-good-join-rate}

In this section, we describe how we estimate {\boldmath{$\iJRate_i$}}, which is intuitively the good join rate in iteration $i$.  More specifically, consider an iteration $i$ that overlaps epochs with indices defined by the set  $\mathcal{E}$. For $j\in \mathcal{E}$, let $\lambda_j$ be the fraction of iteration $i$ that overlaps with epoch $j$; note that $\sum_{j\in \mathcal{E}}\lambda_j = 1$. Then we define {\boldmath{$\iJRate_i$}}$ = \sum_{j\in \mathcal{E}}  \lambda_j \epochRate_i$; that is, the fraction of the overlap multiplied by the good join rate in the respective epoch. We refer to this quantity as the \defn{average good join rate over iteration $i$} or simply ``the good join rate over iteration $i$" (omitting ``average'').

\algGM uses a good join rate estimate, {\boldmath{$\JoinEst_i$}}, in Line $1$ of  Figure~\ref{alg:gmcom} to set the entrance puzzle difficulty.  Computing {\boldmath{$\JoinEst_i$}} is technically challenging, and is done as follows. The committee continually keeps track of sets {\boldmath{$\estSet$}}, which is the most recent system membership,
 where $|\curIDs - \estSet| \geq \frac{3}{5} |\curIDs|$ (Line (i) in Figure~\ref{alg:gmcom}) and {\defn{$\curIDs$}} is the set of good IDs currently in the system. Whenever $\estSet$ is updated --- and only when $\estSet$ is updated --- the parameter {\boldmath{$\estDur$}} is set to the length of time since the system membership was most recently $\estSet$ (Line (ii) in \algGM), and we refer to this length of time as an \defn{interval}. Then, $\JoinEst_i= |\curIDs|/\estDur$ is used as an estimate for the good join rate over iteration $i$.  The setting of this variable is performed in Line 2(d) of \algGM.

We demonstrate in Section~\ref{sec:estimating}  that $\JoinEst_i$ is always within a constant-factor of $\iJRate_i$.  Our approach for estimating the join rate of good IDs may be of independent interest for settings where Assumptions A1 and A2 hold, and it is worthwhile considering why these assumptions are necessary, as discussed below.

Since good and bad IDs cannot be distinguished with certainty, \algGM cannot detect when an epoch begins and ends.  Instead, \algGM finds a lower bound on the number of new good IDs that have joined the system by pessimistically subtracting out the fraction of IDs that could be bad.  When this estimate of new good IDs is sufficiently large --- that is, when the current interval ends ---  \algGM ``guesses'' that at least one  epoch has occurred.  However, multiple (but still a constant number of) epochs may have actually occurred during this time. Thus, for this approach to yield a constant-factor estimate of $\iJRate_i$, one requirement is that the true good join rate cannot have changed ``too much" between these epochs, and Assumption A1 bounds the amount of such change.

Note that Assumption A1 addresses the good join rate measured over entire, consecutive epochs. However, an interval may end {\it within} some epoch $j$. If the good join rate over the portion of epoch $j$ overlapped by the interval deviates by more than a constant factor from $\epochRate_j$, then the estimate will be inaccurate. Assumption A2 ensures that this cannot happen.
	

\setlength\extrarowheight{2pt} 
\newcolumntype{P}[1]{>{\centering\arraybackslash}p{#1}}
\begin{table}[t]
  \begin{center}
    \caption{{A summary of additional notation used for \algGM.} }\vspace{-5pt}
    \label{tab:notation2}
    {\def\arraystretch{1.1}
    {
\begin{tabular}{ |P{1.5cm}|p{10.5cm}|  }
   \hline
 \multicolumn{2}{|c|}{{\bf Used in Model and Algorithm}}\\
     \hline
      \textbf{Symbol} &   \textbf{Description}\\
    \hline
     $\elen_i$ & Length of epoch $i$.\\
               \hline 
      $\epochRate_i$ & The good join rate in epoch $i$; that is, the number of good IDs joining in epoch $i$ divided by $\elen_i$.\\
      \hline
      $\JoinEst_i$ & Estimate of  the good join rate in iteration $i$.\\
       \hline 
      $\iJRate_i$ & Average good join rate over iteration $i$; that is, $\sum_{j\in \mathcal{E}}  \lambda_j \epochRate_i$, where $\mathcal{E}$ is the set of indices of epochs overlapped by the iteration, and $\lambda_j$ is the fraction of the iteration which overlaps epoch $j\in \mathcal{E}$.\\ 
      \hline    
     $\rateCur$ & Current number of join events in present iteration divided by\\
      & current time elapsed in this iteration.\\
      \hline
      $\estSet$ & Most recent system membership where $|\curIDs - \estSet| = (3/5) |\estSet|$ \\ 
      \hline
      $\estDur$ & Amount of time since system membership was most recently $\estSet$.\\
             \hline
             $\advAveCost_{\mathcal{I}}$ &  Cost to the adversary for solving puzzles whose solutions are used in any iteration of $\mathcal{I}$ divided by total length of those iterations.\\
             \hline
             $\joinRate_{\mathcal{I}}$ & Number of good IDs that join over the iterations in $\mathcal{I}$ divided by total length of those iterations. \\
             \hline
     \end{tabular}
    }} \vspace{-15pt}
  \end{center}
\end{table}


\subsection{Overview of \algGM}\label{subsec:overview} 

We now describe the execution of \algGM in iteration $i$ while referring to our pseudocode in Figure~\ref{alg:gmcom}.  

As stated earlier in Section~\ref{subsec:estimate-good-join-rate}, estimating the good-ID join rate involves tracking $\estSet$
and $\estDur$ in Lines (i) and (ii), respectively. The remainder of the iteration executes in two steps.
 
In Step 1, each joining ID must solve a puzzle of difficulty $1$ plus the number of IDs that join within the last $1/\JoinEst_i$ seconds, where $\JoinEst_i$ is our estimate of the join rate for good IDs.   We provide intuition for this cost  in Section~\ref{sec:entrance-cost}.

Step 1 lasts until the earliest point in time when the number of IDs that join in iteration $i$, {\boldmath{$n_i^a$}},  plus the number of IDs that depart in iteration $i$,  {\boldmath{$n_i^d$}}, is at least  $(1/11)|\setIDs_{i-1}|$. The quantities $n_i^a$ and $n_i^d$ are tracked by the committee.

When Step 1 ends, a purge is performed by the committee by issuing a  $1$-hard puzzle via \Diffuse via Step 2(a). In Step 2(b), each ID must respond with a valid solution within $1$ round, or else be removed from the system. The committee removes unresponsive or late IDs from its whitelist, which is maintained via Byzantine consensus amongst the committee members.

The current committee then selects $\Theta(\log n_0)$ IDs uniformly at random from  $\setIDs_{i}$ in Step 2(c).  The committee uses \Diffuse to inform $\setIDs_{i}$ that the selected IDs are the new committee for iteration $i+1$.  All messages from committee members are verified via public key digital signatures.

In Step 2 the estimate of the good-ID join rate over this iteration is calculated and iteration $i$ ends.


\subsection{Intuition for Entrance-Cost Function}\label{sec:entrance-cost}

We now give intuition for our entrance-cost function. Consider iteration $i$.  In the absence of an attack, the entrance cost should be proportional to the good join rate. This is indeed the case since the puzzle difficulty is $O(1)$ corresponding to the number of (good) IDs that join within the last $1/\JoinEst_i  \approx 1/\iJRate_i$ seconds.

In contrast, if there is a large attack, then our entrance-cost function imposes a significant cost on the adversary. Consider the case where a batch of many bad IDs is rapidly injected into the system. This drives up the entrance cost since the number of IDs joining within  $1/\JoinEst_i$ seconds increases.

More precisely, assume the adversary's spending rate is $T = \xi  \joinRateAll_i$, where {\boldmath{$\xi$}} is the entrance cost, and {\boldmath{$\joinRateAll_i$}} is the join rate for all IDs. For the good IDs, the spending rate due to the entrance cost is $\xi \joinRate_i$, and the spending rate due to the purge cost is $\joinRateAll_i$. Setting these to be equal, and solving for $\xi$, we get $\xi  = \joinRateAll_i/ \joinRate_i$;  in other words, the number of IDs that have joined over the last $1/ \joinRate_i$ seconds.  This is the entrance cost function that best balances entrance and purge costs.

Note that spending rate of good IDs due to the entrance costs and purge costs is:
\begin{eqnarray*}
\xi \joinRate_i + \joinRateAll_i & \leq & 2\joinRateAll_i \mbox{\hspace{62pt}By our setting of $\xi$.}\\
& = & 2\sqrt{\left(\joinRateAll_i\right)^2}\nonumber\\
&=& 2 \sqrt{\joinRateAll_i  \xi \joinRate_i} \mbox{\hspace{38pt}Since $\joinRateAll_i = \xi \joinRate_i$.}\\
& = & 2\sqrt{\joinRate_iT} \mbox{\hspace{47pt}Since $T = \xi  \joinRateAll_i$.}
\end{eqnarray*}

This informal analysis provides intuition for the asymmetric cost guaranteed by \algGM.  Proving that $\JoinEst_{i}$ is a good estimate for $\iJRate_i$, showing that the above analysis holds even only with good estimates of join rates, and incorporating multiple epochs and purge costs, is the subject of our analysis in the next section.


\section{Analysis of \algGM} \label{s:anal-gmcomm}
In this section we prove correctness and analyze computational costs for \algGM.

\subsection{Maintaining the Population and Committee Invariants}\label{sec:pop-invariant-gmcom}
\algGM ensures a modified Population Invariant which guarantees that at most a $1/6$-fraction of IDs are bad.

\begin{lemma}\label{lem:pop-gmcom}
For $\alpha\leq 1/18$, \algGM guarantees that the fraction of bad IDs in the system is less than $1/6$.

\end{lemma}
\begin{proof}
By Lemma~\ref{lem:modifiedpop}, if $\alpha = 1/18$ and a threshold for triggering a purge is  $\frac{2\alpha}{1 + 3 \alpha}|\setIDs_{i-1}| = (2/18)/(7/6)|\setIDs_{i-1}| > (1/11)|\setIDs_{i-1}|$, then the fraction of bad IDs in the system is always less than $3\alpha = 1/6$.
\end{proof}

\noindent The threshold for the number of join and departure events for triggering a purge also changes for \algGM, and so we prove that the Committee Invariant still holds.

\begin{lemma}\label{lem:maj_comm-GM} 
With probability $1 - O(1/n_0)$, over a polynomial number of join and departure events, the fraction of good IDs in the committee is always at least $10/11 - \alpha - 2 \epsilon_d$. \end{lemma}
\begin{proof}
The argument is almost identical to the proof of Lemma~\ref{lem:maj_comm}, but we provide the details for completeness.  For iteration  $i = 0$, the committee invariant holds by the use of \cite{hou2017randomized} to initialize the system (recall Section~\ref{sec:overview}; for details, see Lemma 6 of \cite{andrychowicz2015pow}).
	
Fix an iteration $i > 0$. Recall that a new committee is elected by the existing committee, by selecting $c\log |S_i|$ IDs independently and uniformly at random from the set $S_i$ of IDs left after the purge, where $c>0$ is a sufficiently large constant set later in this proof. Let $X_G$ be a random variable giving the number of good IDs in the committee at formation. Then:
	\begin{eqnarray}\label{eq:expgc-gm}
		E[X_G] & = & \frac{|G_i|}{|S_i|}c\log{|S_i|} \text{ } = \text{ } (1-\alpha)c\log|S_i| 
	\end{eqnarray}
	 where the last inequality follows from the fact that the fraction of computational power with the adversary is at most $\alpha$.  By Chernoff bounds~\cite{mitzenmacher2017probability}, for any $\delta$, $0 < \delta < 1$:
	 \begin{eqnarray*}
	 	Pr \left( X_G < (1-\delta)(1-\alpha)c\log|S_i| \right) \leq  \exp{\left\{ -\frac{\delta^2 (1-\alpha) c\log|S_i|}{2} \right\}} 
 \end{eqnarray*}
 
For any constant $\epsilon_0 > 0$, setting $\delta = \frac{\epsilon_0}{1-\alpha}$, we have that:
 $$ Pr(X_g < (1- \alpha - \epsilon_0) c\log{|S_i|}) = O(n_0^{-(\gamma + 1)}).$$
 
Where the last equality holds for $c \geq \frac{28}{13}\frac{(\gamma + 1)}{\delta^2}$.
 	 	
Next, let $Y_g$ be a random variable denoting the number of good IDs that depart from the committee during iteration $i$, but before the purge, i.e. up to the point when the number of departures of IDs from the system is $|S_{i-1}|/11$.  Departing good IDs are selected uniformly at random from the population of good IDs (See Section~\ref{sec:join}). Thus, we obtain:
		\begin{eqnarray}\label{eq:dgoodc-gm}
			E[Y_g] & \leq & \frac{|S_i|}{11}\left(\frac{c\log|S_i|}{|S_i|}\right) \text{ } = \text{ } \frac{c}{11} \log|S_i|
		\end{eqnarray}  
		By Chernoff bounds~\cite{mitzenmacher2017probability}:
		\begin{eqnarray*}
			Pr \left( Y_g > {(1+\delta')} \frac{c \log|S_i|}{11} \right) &\leq&  \exp\left\{ -\frac{\delta'^2c}{33}\log|S_i| \right\} = O\left( {n_0^{-(\gamma+1)}} \right)
		\end{eqnarray*}
		where the first inequality holds for any $\delta$, $0< \delta' < 1$ and the equality holds for all $c \geq \frac{33(\gamma + 1)}{\delta'^2}$. 

Thus, for any fixed $\epsilon_1>0$, the minimum number of good IDs in the committee at any point during the iteration is greater than $(1 - \alpha - 1/11 - \epsilon_0 - \epsilon_1) c\log{|S_i|}$. Good IDs can also depart while the committee is running the single-round purge at the end of iteration.  But, by assumption, at most an $\epsilon_d$-fraction of good IDs depart per round, where $\epsilon_d$ is sufficiently small. Setting $\epsilon_0 = \epsilon_1 = (1/2) \epsilon_d$, we get that the number of good IDs exceeds $(10/11 - \alpha - 2 \epsilon_d) c\log{|S_i|}$, with probability of error $O\left( {n_0^{-(\gamma+1)}} \right)$. 

Finally, the claim follows by taking a union bound over all $n_0^\gamma$ iterations.
\end{proof}


\subsection{Estimating the Good ID Join Rate}\label{sec:estimating}
 
To calculate the entrance cost in iteration $i$, \algGM requires an estimate of the good ID join rate, $\iJRate_{i}$.  However,  good IDs cannot be discerned from bad IDs upon entering the system,  and so the adversary may inject bad IDs in an attempt to obscure the true join rate of good IDs. 

In this section, we prove that our estimate of the good-ID join rate, $\JoinEst_i$,  is within a constant factor of $\iJRate_i$. For any time {\boldmath{$t$}}, let {\boldmath{$\setIDs_t$}} and {\boldmath{$\setGood_t$}} denote the set of all IDs and set of good IDs, respectively, in the system at time $t$. 

We say that interval $\estDur$ \defn{touches} an epoch if there is a point in time belonging to both the interval $\estDur$ and the interval of time corresponding to the epoch; it does not necessarily mean that $\ell$ completely contains the epoch, or vice versa.

\begin{lemma}\label{lem:iteration-epochs}
An iteration touches at most $2$ epochs.
\end{lemma}
\begin{proof}
Assume that some iteration $i$ starts at time $t_0$ and touches at least three epochs; we will derive a contradiction. This assumption implies that there is at least one epoch entirely contained within iteration $i$.  Let this epoch start at time $t_1\geq t_0$ and end at time $t_2 > t_1$ within iteration $i$. 

By definition, the number of good IDs that join over this epoch is at least $(3/4)|\setGood_{t_2}| $.   Observe that:
 \begin{eqnarray*} 
 (3/4)|\setGood_{t_2}|  & \geq & (3/4) (5/6)|\setIDs_{t_2}| \mbox{~~~by the Population Invariant (Lemma~\ref{lem:pop-gmcom})}\\
& = & (5/8)|\setIDs_{t_2}|\\
& > & (5/8)((10/11)|\setIDs_{t_0}|)\\
& = & (25/44)|\setIDs_{t_0}| 
\end{eqnarray*}

\noindent where the second to last line follows from noting that $|\setIDs_{t_2}| > |\setIDs_{t_0}| - (1/11)|\setIDs_{t_0}| \geq  (10/11)|\setIDs_{t_0}| $, since at most $(1/11)|\setIDs_{t_0}|$ departures can happen before a purge occurs.

Finally, observe two bounds.  First,  during any iteration at most $(1/11)|\setGood_{t_0}|$  good IDs join before the purge. Second,  
at most $\epsilon_a(1 + 1/11) |\setGood_{t_0}| = \epsilon_a(12/11) |\setGood_{t_0}|$ join during the single round  of the purge. Therefore, the total number of good IDs that can join over an iteration is $(1/11)|\setGood_{t_0}| + \epsilon_a(12/11) |\setGood_{t_0}|$.
The contradiction follows from noting that:
 $$(25/44)|\setIDs_{t_0}| > (1/11)|\setGood_{t_0}| + \epsilon_a(12/11) |\setGood_{t_0}|.$$
\noindent for $\epsilon_a < 7/16$. 
\end{proof}

\begin{lemma}\label{lem:interval-epochs}
An interval touches at most two epochs and cannot completely overlap any single epoch.
\end{lemma}
\begin{proof}
Assume that some interval starts at time $t_0$ and touches at least three epochs; we will derive a contradiction. This assumption implies that there is at least one epoch entirely contained within the interval.  Consider the first such epoch, and let it start at time $t_1\geq t_0$ and end at time $t_2> t_1$.

Observe that:
 \begin{eqnarray*} 
 |\setIDs_{t_2}  -\setIDs_{t_1} |  &\geq &| \setGood_{t_2} - \setGood_{t_1}| \\
 & \geq & (3/4)| \setGood_{t_2}| \mbox{~~~by the definition of an epoch}\\
&\geq & (3/4) (5/6)|\setIDs_{t_2}| \mbox{~~~by the Population Invariant (Lemma~\ref{lem:pop-gmcom})}\\
& = & (5/8)|\setIDs_{t_2}|\\
& > & (3/5) |\setIDs_{t_2}| 
\end{eqnarray*}
But this is a contradiction since it implies that the interval must end before time $t_2$.
\end{proof}

\begin{figure*}[h!]
\centering
\includegraphics[width = 0.9\textwidth]{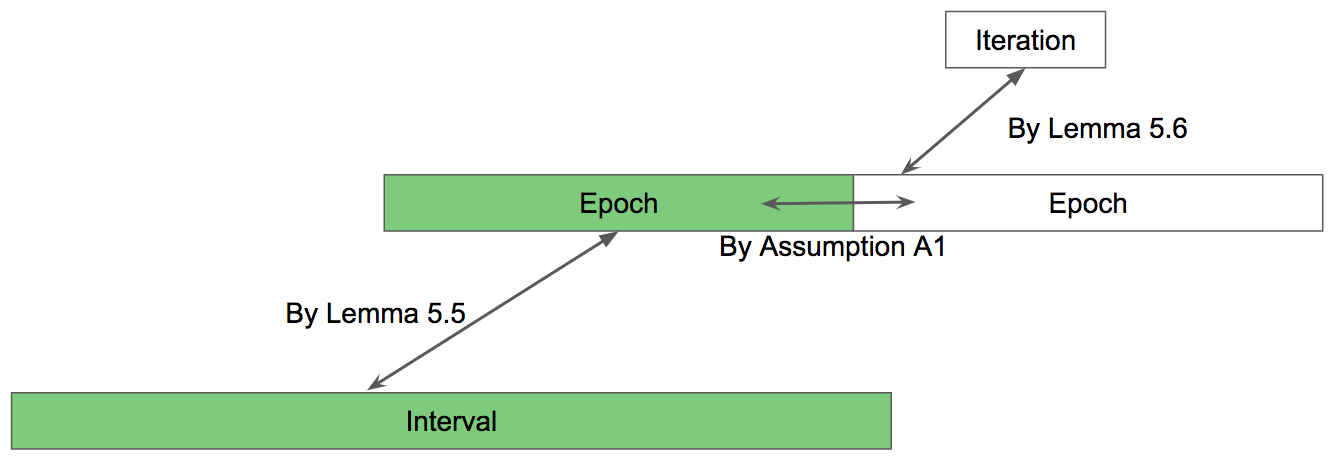}
\caption{A depiction of the relationship between the good join rate in iterations, epochs, and durations. Arrows represent that good join rates of the endpoints are within constant factors of each other. Green indicates the most recently-finished epoch and duration being tracked by the committee to obtain $\JoinEst_i$.}
  \label{fig:estimating-JG}  
  \vspace{0pt}
\end{figure*}

\begin{lemma}\label{lem:bounded}
Consider any interval of length $\estDur$, where there are $|\curIDs|$ IDs in the system at the end of the interval, and let epoch $i\geq 1$ be the most recent epoch that the interval touches. Then:
\end{lemma}
$$\left(\frac{5}{6}\right) \AOneL \ATwoL \epochRate_i  \leq \frac{|\curIDs|}{\estDur} \leq 5 \AOneH \ATwoH \epochRate_i.$$
\begin{proof}
We break the argument into three cases.\smallskip

\noindent{\bf Case 1:} The interval touches only a single epoch where the epoch begins at $t_0$, the interval begins at $t_1$, and the interval ends at $t_2$. By Assumption A2, we have:

$$\ATwoL \epochRate_i \leq \frac{|\setGood_{t_2} - \setGood_{t_1}|}{t_2 - t_1} \leq   \ATwoH \epochRate_i$$

\noindent and by the Population Invariant and the specification of an interval: 
$$\left(\frac{2}{5}\right)|\setIDs_{t_2}| \leq \left(\frac{3}{5} - \frac{1}{6}\right)|\setIDs_{t_2} | \leq |\setGood_{t_2} - \setGood_{t_1}| \leq \left(\frac{3}{5}\right)|\setIDs_{t_2} |. $$

\noindent Therefore, we have:

$$ \left(\frac{5}{3}\right) \ATwoL \epochRate_i \leq \frac{ |\setIDs_{t_2} |}{t_2 - t_1} \leq \left(\frac{5}{2}\right) \ATwoH \epochRate_i $$

\smallskip

\noindent{\bf Case 2:}  The interval touches epochs $i-1$ and $i$ and there are at least $2$ good ID join events in each epoch. Let epoch $i-1$ start at time $t_0$ and end at $t_2$ (and so epoch $i$ starts at $t_2$), and let the interval start at time $t_1 \geq t_0$ and end at $t_3 \geq t_2$. By Assumptions A1 and A2, we have:

$$ \AOneL \ATwoL \epochRate_i \leq \frac{|\setGood_{t_3} - \setGood_{t_1}|}{t_3 - t_1} \leq  \AOneH \ATwoH \epochRate_i $$

\noindent and by the Population Invariant and the specification of an interval:

$$\left(\frac{2}{5}\right)|\setIDs_{t_3}| \leq \left(\frac{3}{5} - \frac{1}{6}\right)|\setIDs_{t_3} | \leq |\setGood_{t_3} - \setGood_{t_1}| \leq \left(\frac{3}{5}\right)|\setIDs_{t_3} |. $$

\noindent Therefore, we have:

\begin{eqnarray}
\left(\frac{5}{3}\right) \AOneL \ATwoL \epochRate_i  \leq \frac{ |\setIDs_{t_3} |}{t_3 - t_1} \leq  \left(\frac{5}{2} \right) \AOneH \ATwoH \epochRate_i     \label{eq:case2} 
\end{eqnarray}

\noindent{\bf Case 3:} The interval touches epochs $i-1$ and $i$, and w.l.o.g. we have a single good join event in the portion of epoch $i-1$ that overlaps the interval; denote the length of this overlap by $\lambda'>0$.  As with Case 2, let epoch $i-1$ start at time $t_0$ and end at $t_2$ (and so epoch $i$ starts at $t_2$), and let the interval start at time $t_1 \geq t_0$ and end at $t_3 \geq t_2$. 

Observe that the single good join event in epoch $i-1$ increases the numerator of the bounded quantity in Equation~\ref{eq:case2} by $1$, and so twice the upper bound in Case 2 suffices here.  The denominator increases by $\lambda'$, where $\lambda' \leq t_3-t_1$, so half the lower bound in Case 2 suffices. This implies:
\begin{eqnarray*}
\left(\frac{5}{6}\right) \AOneL \ATwoL \epochRate_i  \leq \frac{ |\setIDs_{t_3} |}{t_3 - t_1} \leq  5 \AOneH \ATwoH \epochRate_i    
\end{eqnarray*}
\end{proof}

\begin{lemma}\label{lem:four}
For any epoch $i\geq 1$, and any iteration $i' \geq 1$ that epoch $i$ touches: 
\end{lemma}
$$\AOneL \iJRate_{i'} \leq \epochRate_{i} \leq \AOneH \iJRate_{i'}$$

\begin{proof}
\noindent By Lemma~\ref{lem:iteration-epochs}, iteration $i'$ touches at most two epochs, say epochs $i-1$ and $i$. Thus, by definition of $\iJRate_{i'}$, for some $\lambda_1, \lambda_2 \geq 0$, $\lambda_1+\lambda_2 = 1$:
\begin{eqnarray*}
\iJRate_{i'} & = & \lambda_{1} \epochRate_{i-1} + \lambda_{2} \epochRate_{i}\\
&\leq &  \lambda_{1} \left(\frac{\epochRate_{i}}{ \AOneL}\right) + \lambda_{2} \epochRate_{i}\\
&\leq &   \frac{\epochRate_{i}}{ \AOneL}
\end{eqnarray*}

The second line in the above holds by Assumption A1.  The last line holds since $\lambda_1 + \lambda_{2} = 1$ and  $\AOneL\leq 1$.  A similar derivation yields that $\iJRate_{i'}  \geq  \frac{\epochRate_{i}}{\AOneH}$. Together this implies that:
$$ \AOneL \iJRate_{i'} \leq \epochRate_{i} \leq \AOneH \iJRate_{i'}.$$
\end{proof}

\begin{theorem} \label{t:JoinEst}
For any iteration $i\geq 1$, the following bound holds:
\end{theorem}
$$\left(\frac{5}{6}\right) \frac{\AOneL^2 \ATwoL}{\AOneH}  \iJRate_{i} \leq\JoinEst_i \leq \frac{5\, \AOneH^2 \ATwoH}{\AOneL}  \iJRate_{i}$$
\begin{proof}
The estimate $\JoinEst_i$ used in iteration $i$ corresponds to the most recent interval that completed before iteration $i$ started. Let epoch $j$ be the most recent epoch that touches this interval. By Lemma~\ref{lem:bounded}:
$$\left(\frac{5}{6}\right) \AOneL \ATwoL \epochRate_j \leq\JoinEst_i \leq 5\,\AOneH \ATwoH \epochRate_j.$$

Epoch $j$ may end prior to the start of iteration $i$; that is, epoch $j$ may not necessarily touch iteration $i$. In this case, note that by Lemma~\ref{lem:interval-epochs}, the current interval touches epoch $j+1$ and must end before epoch $j+1$ ends. This fact, along with the observation that the current interval touches iteration $i$, implies that epoch $j+1$  touches iteration $i$. 

By the above, we know that either epoch $j+1$ or epoch $j$ touches iteration $i$.. Lemma~\ref{lem:four} implies:
$$ \AOneL \iJRate_{i} \leq \epochRate_{j} \leq \AOneH \iJRate_{i}$$
or
$$ \AOneL \iJRate_{i} \leq \epochRate_{j+1} \leq \AOneH \iJRate_{i}.$$

To employ our top-most equation, we use A1 to derive:
$$ \epochRate_{j+1}/\AOneH \leq \epochRate_{j} \leq  \epochRate_{j+1}/\AOneL  $$
\noindent and then plug into our top-most equation, we have:
$$\left(\frac{5}{6}\right) \frac{\AOneL^2 \ATwoL}{\AOneH}  \iJRate_{i} \leq\JoinEst_i \leq \frac{5\,\AOneH^2 \ATwoH}{\AOneL}  \iJRate_{i}$$
\end{proof}


\subsection{Cost Analysis} \label{sec:cost-analysis-gm}

We make use of the following two algebraic facts that both follow from the Cauchy-Schwartz inequality.

\begin{lemma}\label{l:cs2}
Let $n$ be a positive number, and for all $1 \leq i \leq n$, $s_i \geq 0$ and let $S = \sum_{i=1}^n s_i$.  Then
$$ \sum_{i=1}^n s_i^2 \geq S^2/n$$
\end{lemma}
\begin{proof}

Let $u$ be a vector of length $n$ with for all $1 \leq i \leq n$, $u[i] = s_i$, and let $v$ be a vector of length $n$ with, for all $1 \leq i \leq n$, $v[i] = 1$.  Then by Cauchy-Schwartz:
\begin{eqnarray*}
	|\langle u, v \rangle|^2 & \leq & \langle u,u \rangle \cdot \langle v, v \rangle \\
	S^2  & \leq & \left(\sum_{i=1}^n s_i^2 \right) \cdot n
\end{eqnarray*}
Rearranging completes the proof.
\end{proof}

Let {\boldmath{$\advCost_i$}} denote the cost rate paid by the adversary  during iteration $i$. Let {\boldmath{$\joinRateBad_{i}$}} be the join rate of bad IDs during iteration $i$. For simplicity, we provide notation for the following constants:
$${\boldmath{\bigconstant}} = \frac{5\,\AOneH^2 \ATwoH}{\AOneL}$$
\noindent and 
$$\smallconstant =  \left(\frac{5}{6}\right) \frac{\AOneL^2 \ATwoL}{\AOneH} $$
\noindent where \texttt{JE} stands for ``join estimate''.

\begin{lemma}\label{l:joinBad}
For any iteration $i>1$, $\joinRateBad_i \leq c(L8) \sqrt{T_i (J^G_i+1)}$, where $c(L8) = \sqrt{2\bigconstant}$.
\end{lemma}

\begin{proof}

For simplicity, we normalize time units so that $\ell_i=1$.  Partition iteration $i$ from left to right into sub-iterations, all of length $1/\JoinEst_i$, except the last, which is of length at most $1/\JoinEst_i$.  We lower bound the cost paid by the adversary for joins by pessimistically assuming that only bad IDs are counted when computing entrance costs. For  $1 \leq x \leq \lceil \JoinEst_i  \rceil$, let $j_x$ be the total number of bad IDs that join in sub-iteration $x$.  Since $\sum_{y=1}^{j_x} y = (j_x + 1)j_x/2 \geq (j_x)^2/2$,  the total entrance cost paid by bad IDs is at least $(1/2) \sum_{x = 1}^{\lceil \JoinEst_i  \rceil} (j_x)^2$. 
 
Since $\sum_{x = 1}^{\lceil \JoinEst_i  \rceil} j_x = \joinRateBad_i$,   by applying Lemma~\ref{l:cs2}, we have:

$$T_i \geq \frac{1}{2} \sum_{x = 1}^{\lceil \JoinEst_i  \rceil} (j_x)^2 \geq \frac{(\joinRateBad_i)^2}{2 (\lceil \JoinEst_i  \rceil)} $$	

Cross-multiplying and taking the square root, we get:

\begin{eqnarray*}
	\joinRateBad_i & \leq & \sqrt{2 T_i \lceil \JoinEst_i  \rceil} \\
	 & \leq & \sqrt{2 T_i (\JoinEst_i +1)} \\
	 & \leq & \sqrt{2  T_i \left(\bigconstant \iJRate_{i} + 1\right)} \\
\end{eqnarray*}
where the second step follows from noting that $\lceil x \rceil \leq x +1$ for all $x$, and the final step follows from Lemma~\ref{t:JoinEst} which states that: 
$$\JoinEst_i  \leq  \frac{ 5 \AOneH^2 \ATwoH}{\AOneL}  \iJRate_{i} = \bigconstant \iJRate_{i} $$
\noindent which yields the lemma statement. 
\end{proof}

\begin{lemma} \label{l:algCost}
		Let $\mathcal{A}_i$ be the average spend rate for the algorithm in any iteration $i>1$.  Then: 
$$ \mathcal{A}_i \leq \frac{c(L9)|S_{i-1}|}{\ell_i}$$
where $c(L9) = \left( \frac{12}{11} + \frac{\AOneH \ATwoH}{11\smallconstant} \right)$.
\end{lemma}
\begin{proof}
For simplicity, we first normalize time units so that $\ell_i=1$.  Partition iteration $i$ from left to right into sub-iterations, all of length $1/\joinRate_i$, except the last, which is of length at most $1/\joinRate_i$.  

The spend rate for the algorithm due to purge costs is  $|S_{i-1}|$.  

For entrance costs, note the following two facts.  First, by Assumptions A1 and A2, there are at most $\AOneH \ATwoH $ good IDs in any sub-iteration (note that a sub-iteration might span two epochs). Second,  the entrance cost for any good ID is $1$ plus the number of join events over the past $1/\JoinEst_i$ seconds. By Theorem~\ref{t:JoinEst}, $1/\JoinEst_i \leq 1/(\smallconstant\joinRate_i)$, and so the entrance cost is at most $1$ plus the number of join events over the past $1/\smallconstant$ sub-iterations; that is $1 + \AOneH \ATwoH/\smallconstant$.

By these two facts, and by the fact that there are at most $|S_{i-1}|/11$ join events in an iteration, the total entrance costs paid by good IDs in iteration $i$ is at most:
$$\left(1 + \frac{\AOneH \ATwoH}{\smallconstant}\right)\left(\frac{|S_{i-1}|}{11}\right).$$  

Adding the bounds for both entrance and purge costs and dividing by the value of $\ell_i$ yields a total cost of at most:
\begin{eqnarray*}
&&|S_{i-1}| + \left(1 + \frac{\AOneH \ATwoH}{\smallconstant}\right)\left(\frac{|S_{i-1}|}{11}\right)\\  
&= & \frac{|S_{i-1}|}{\ell_i}\left( \frac{12}{11} + \frac{\AOneH \ATwoH}{11\smallconstant} \right)
\end{eqnarray*}
\noindent which completes the proof. 
\end{proof}

\begin{lemma}\label{l:cs1}
	Let $u$ and $v$ be $n$-dimensional vectors. Then $$\sum_{j=1}^n\sqrt{u_jv_j} \leq \sqrt{\left(\sum_{j=1}^n u_j\right) \left(\sum_{j=1}^n v_j\right)}.$$
\end{lemma}
\begin{proof}
	Let $x$ and $y$ be vectors such that for all $1 \leq i \leq n$, $x_i = \sqrt{u_i}$ and $y_i = \sqrt{v_i}$.  Then applying the Cauchy-Schwarz inequality to $x$ and $y$, we have that $x \cdot y \leq |x||y|$, from which the lemma follows directly.
\end{proof}

Let $\Iters$ be any subset of iterations that for integers $x$ and $y$, $1 \leq x \leq y$, contains every iteration with index between $x$ and $y$ inclusive.  Let $\delta(\Iters)$ be $|S_{x} - S_y|$; and let $\Delta(\Iters)$ be  $\delta(\Iters)$ divided by the length of $\Iters$. Let $\mathcal{A}_{\Iters}$ and $T_{\Iters}$ be the algorithmic and adversarial spend rates over $\Iters$; and let $\iJRate_{\Iters}$ be the good join rate over all of $\Iters$.  Then we have the following lemma.

\begin{lemma}\label{l:cost}
   	For any subset of contiguous iterations, $\Iters$, which starts after iteration $1$, it holds that:
\end{lemma}
$$11c(L9)  \left(  2\Delta(\Iters)  +  c(L8)\sqrt{2T_{\Iters} (\bigconstant \iJRate_{\Iters} + 1)}  + \iJRate_{\Iters}\right).$$
\begin{proof}
By Lemma~\ref{l:algCost}, we have:  
\begin{equation} 	  \label{eqn:algcost}
	\sum_{i \in \Iters} \mathcal{A}_i \ell_i \leq c(L9)\sum_{i \in \Iters} |S_{i-1}| 
\end{equation}
\noindent Let {\boldmath{$\depRateTot_{i}$}} be the departure rate of both good and bad IDs during iteration $i$.  By the conditions for when a purge occurs, we have that:
\begin{eqnarray*}
	\sum_{i \in \Iters} \frac{|S_{i-1}|}{11} & \leq & \sum_{i \in \Iters} (\depRateTot_i + \joinRateBad_i + \iJRate_{i}) \ell_i \\
	& \leq & 2 \delta(\Iters) + \sum_{i \in \Iters} \joinRateBad_i \ell_i + \sum_{i \in \Iters}\iJRate_{i}\ell_i   \\
	& \leq & 2 \delta(\Iters) + \sum_{i \in \Iters} \sqrt{2T_i \ell_i (\bigconstant \iJRate_i + 1) \ell_i} + \sum_{i \in \Iters} \iJRate_{i} \ell_i \\
    & \leq & 2 \delta(\Iters) +  \sqrt{\left( \sum_{i \in \Iters} 2T_i \ell_i \right) \left( \sum_{i \in \Iters} (\bigconstant \iJRate_i + 1) \ell_i \right)} + \sum_{i \in \Iters} \iJRate_{i} \ell_i 
\end{eqnarray*}

The second line in the above follows from the fact that every ID that departs must have departed from the set of IDs in the system at the start of $\Iters$ or else must have been an ID that joined during $\Iters$.  The third line follows from the Lemma~\ref{l:joinBad} bound on $\joinRateBad_i$, and noting that $\ell_i = \sqrt{\ell^2_i}$.  Finally, the last line follows from Lemma~\ref{l:cs1}, by letting $u_i = 2T_i\ell_i$ and $v_i = (\bigconstant \joinRate_i + 1)\ell_i$ for all $i \in \mathcal{I}$.

Now combining Equation~\ref{eqn:algcost} with the above inequality, we have:
\begin{eqnarray*}
	\sum_{i \in \Iters} \mathcal{A}_i \ell_i & \leq & 11 c(L9) \left( 2 \delta(\Iters) +  \sqrt{\left( \sum_{i \in \Iters} 2T_i \ell_i \right) \left( \sum_{i \in \Iters} (\bigconstant\iJRate_i + 1) \ell_i \right)} + \sum_{i \in \Iters} \iJRate_{i} \ell_i \right)
\end{eqnarray*}

Dividing both sides of the above inequality by $\ell_{tot} = \sum_{i \in \Iters} \ell_i$ and recalling that $c(L8) = \sqrt{2\bigconstant}$, we get
\begin{eqnarray*}
	 \mathcal{A}_{\Iters} & \leq & 11c(L9)\left( \frac{2 \delta(\Iters)}{\ell_{tot}} +  c(L8)\sqrt{\frac{\left( \sum_{i \in \Iters} T_i \ell_i \right)}{\ell_{tot}} \frac{\left( \sum_{i \in \Iters} (\iJRate_i + 1) \ell_i \right)}{\ell_{tot}}} + \frac{\sum_{i \in \Iters} \iJRate_{i} \ell_i}{\ell_{tot}} \right)\\
	  & = & 11c(L9)  \left(  2\Delta(\Iters)  +  c(L8)\sqrt{2T_{\Iters} (\bigconstant \iJRate_{\Iters} + 1)}  + \iJRate_{\Iters}\right)
\end{eqnarray*}
\noindent which completes the proof.
\end{proof}
 
Maintenance of the Population Invariant is guaranteed by Lemma~\ref{lem:pop-gmcom}, and the Committee Invariant is guaranteed by Lemma~\ref{lem:maj_comm-GM}.  Corollary~\ref{cor:main-upper} follows directly from Lemma~\ref{l:cost}. Theorem~\ref{thm:main-upper} also follows from Lemma~\ref{l:cost} by  noting that $\Delta(\Iters)=0$ when $\Iters$ is all iterations, since the system is initially empty.

Finally, in concluding this section, we note that the same enhancements to \AlgB presented in Section~\ref{sec:enhancements} can be applied to \algGM. \medskip

\section{Lower Bounds}\label{sec:lower}

In this section, we provide a lower bound that applies to the class of algorithms which have the attributes $B1 - B3$ described below. 

\smallskip

\begin{itemize}
\item{\bf B1.} Each new ID must pay an entrance fee in order to join the system and this is defined by a \defn{cost function} {\boldmath{$f$}}, which takes as inputs the good join rate and the adversarial join rate.\smallskip
\item{\bf B2.} The algorithm executes over iterations, delineated by when the condition $n_i^a + n_i^d \geq  \delta\,|\mathcal{S}_i|$ holds, for any fixed positive $\delta$.\smallskip
\item {\bf B3.} At the end of each iteration, each ID must pay $\Omega(1)$ to remain in the system.
\end{itemize}

\smallskip
We emphasize that B1 captures any cost function where the cost during an iteration depends only on the good join rate and the bad join rate.  Our analysis of \algGM makes this assumption since its cost function depends on estimates of the good join rate and the total, i.e. the good-ID join rate plus the bad-ID join rate.  With regard to B2 and B3, recall that we wish to preserve the population invariant (i.e., a majority of good IDs).  It is hard to imagine an algorithm that preserves this invariant without a computational test being imposed on all IDs whenever the system membership changes significantly.

\subsection{Lower-Bound Analysis}
\noindent Restating in terms of the conditions above, we have the following result.\smallskip

{\bf Theorem 2.}
{\it Suppose an algorithm satisfies conditions B1-B3, then there exists an adversarial strategy that forces the algorithm to spend at a rate of $ \Omega(\sqrt{\advAveCost\,\joinRate} + \joinRate)$, where $\joinRate$ is the good ID join rate, and $\advAveCost$ is the algorithmic spending rate, both taken over the iteration.
}

\begin{proof}
 Fix an iteration $i$.  Let $n$ be the number of IDs in the system at the start of iteration $i$.  The adversarial will have bad IDs join uniformly at the maximum rate possible, and then have the bad IDs drop out during the purge.  In particular, let $\advRate$ be the rate at which bad IDs join, and let $f(\advRate, \joinRate)$ be the algorithm's entrance cost function based on $\advRate$ and $\joinRate$; we pessimistically assume that the algorithm knows both $\advRate$ and $\joinRate$ exactly.  Then $\advRate = T / f(\advRate, \joinRate)$.

We first calculate the algorithmic spending rate due to purge puzzle costs in iteration $i$.  Since the iteration ends after $\Theta(n)$ join events (B2), and since each purge puzzle has a cost of $\Omega(1)$ (B3), the rate of spending on purge puzzles is asymptotically at least equal to the rate at which good and bad IDs join the system.  Thus, the spending rate due to purge puzzles is $\Omega (\joinRate + \advRate)$.
 		 	
\smallskip			
\noindent We now have two cases:\medskip

\noindent	
     \textbf{Case 1:} {\boldmath{$f(\advRate, \joinRate) \leq \advRate/\joinRate$}}.  In this case, we have:    
           \begin{eqnarray*}
	 	\advRate & \geq & \joinRate f(\advRate, \joinRate) =  \joinRate T/ \advRate
	 \end{eqnarray*}
where the above equality holds since $f(\advRate, \joinRate) = \advAveCost/\advRate$.  Solving for $\advRate$ we get:
	 \begin{eqnarray*}
	 	\advRate &\geq & \sqrt{\advAveCost\,\joinRate}
	 \end{eqnarray*} 	 

Since the algorithmic spending rate due to purge costs is $\Omega(\advRate+ \joinRate)$, the total algorithmic spending rate is:
      \begin{eqnarray*}
 \Omega(\advRate + \joinRate) & = & \Omega\left(\sqrt{\advAveCost\,\joinRate} + \joinRate\right)
	 \end{eqnarray*}

\noindent\textbf{Case 2:} {\boldmath{$f(\advRate, \joinRate) > \advRate/\joinRate$}}. In this case, we have:	 
	 \begin{eqnarray*}
	 	\advRate & < & \joinRate\,f(\advRate, \joinRate) = \joinRate T/\advRate
	 \end{eqnarray*}
where the above equality follows since $f(\advRate, \joinRate) = \advAveCost/\advRate$.  Solving for $\advRate$, we get:
	 \begin{eqnarray*}
	 	\advRate & < & \sqrt{\advAveCost\,\joinRate}.
	 \end{eqnarray*}

The spending rate for the algorithm due to entrance costs is $\Omega(\joinRate\, f(\advRate, \joinRate))$.  Adding in the spending rate for purge costs of $\Omega (\joinRate + \advRate)$, we get that the total spend rate of the algorithm is:
	 \begin{eqnarray*}
 \Omega(\joinRate\, f(\advRate, \joinRate) + (\joinRate + \advRate))
	 & = & \Omega(\joinRate\, \advAveCost/\advRate + (\joinRate + \advRate)) \\
	 & = & \Omega\left(\sqrt{\advAveCost\,\joinRate} + \joinRate\right)
	 \end{eqnarray*}
\end{proof}



\section{Experiments}~\label{sec:experiments} 

\noindent In this section, we report our simulation results for \AlgB and \algGM. Our goal is to examine the performance predicted by our theoretical results with respect to the computational cost of solving puzzles. Given this goal, we do not model Byzantine consensus or committee formation.  Throughout all of our experiments, we assume a computational cost of $k$ for solving a puzzle of difficulty $k$. 

We present our findings as follows. In Section \ref{s:JandLAssum}, we validate our assumptions from Section \ref{sec:modelGMCom} over real-world peer-to-peer networks. In Section \ref{section:empasym}, we evaluate the computational cost for the proposed algorithms, \AlgB and \algGM, as a function of the adversarial cost, and we compare our algorithms against another PoW-based Sybil Defense algorithm.  Finally,  in Section \ref{section:heuristics}, we propose and implement several heuristics to optimize the performance of \algGM.

Our experiments make use of data from four real-world networks: 
\begin{itemize}
\item{\it Bitcoin:}  This dataset consists of roughly $7$ days of join/departure-event timestamps; this data was obtained by personal correspondence with Till Neudecker~\cite{7140490}. The dataset records the join and departure events of IDs in the network, timestamped to the second.  
\item {\it BitTorrent Debian:} Peers connecting to the BitTorrent network to obtain a Debian ISO image.
\item {\it BitTorrent RedHat:} Peers connecting to the BitTorrent network to obtain a RedHat ISO image.
\item {\it BitTorrent FlatOut:} Peers connecting to the BitTorrent network to obtain aa demo of the game Flatout.
\end{itemize} 

For the BitTorrent data, our simulations assume the Weibull distribution to model session time, with shape parameter values $0.38$, $0.59$ and $0.34$, and scale parameter values $42.2$, $41.0$ and $21.3$; these values are chosen based on a well-cited study by Stuzbach and Rejaie~ \cite{Stutzbach:2006:UCP:1177080.1177105}.

Finally, we note that our simulation code is written in MATLAB and can be obtained from~\cite{diksha-code}. All of our experiments were executed on a MacBook Air with macOS Mojave (version 10.14.3) using a 1.6 GHz Intel Core i5 processor and 8 GB of 2133 MHz LPDDR3. 

\subsection{Testing \algGM Join and Departure Assumptions} \label{s:JandLAssum}

\begin{figure*}[b!]
\vspace{-10pt}
\captionsetup[subfigure]{labelformat=empty}
\begin{subfigure}{0.49\textwidth}
	\centering
	\includegraphics[trim = 1cm 6cm 1cm 6.5cm, width=1\textwidth]{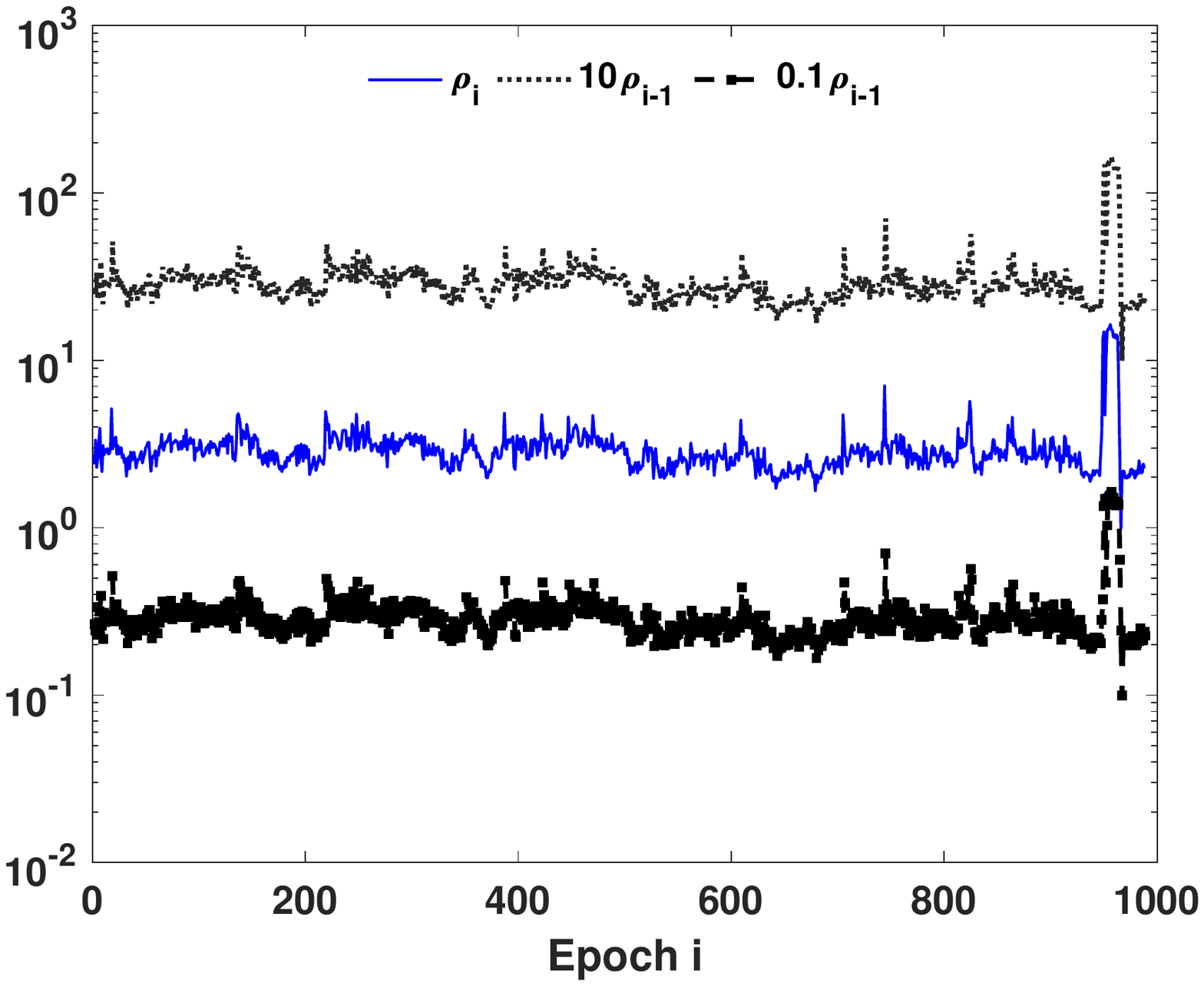} 
	\vspace{-15pt}
	\caption{(a)}
\end{subfigure}
\begin{subfigure}{0.49\textwidth}
	\includegraphics[trim = 1cm 6cm 1cm 6.5cm, width=1\textwidth]{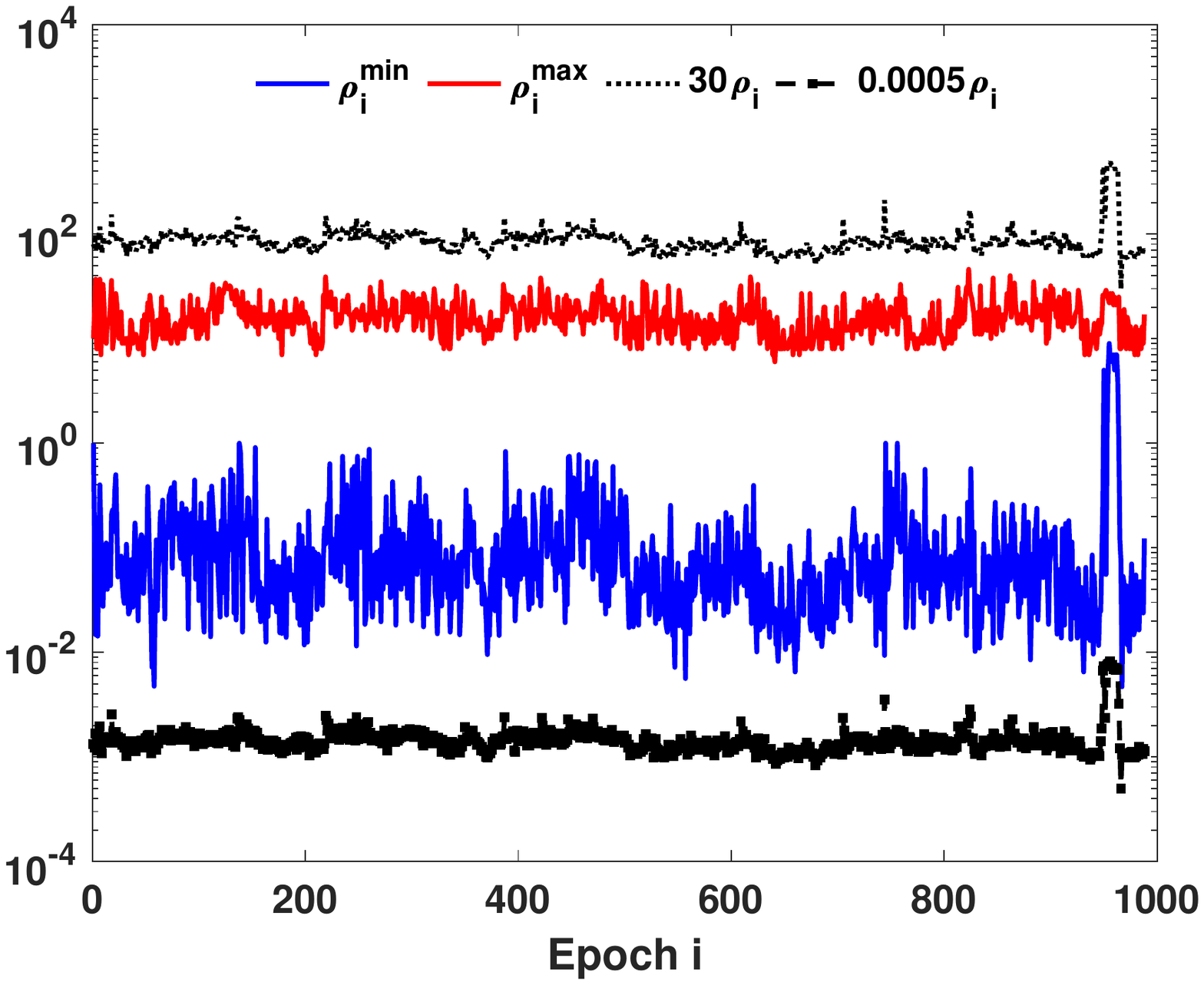}  
	\vspace{-15pt}
	\caption{(b)}
\end{subfigure}
\setcounter{figure}{5} 
\caption{Testing Assumptions A1 and A2 for Bitcoin.}
\label{fig:BitcoinAssumptions}
\end{figure*}

\begin{figure*}[h!]
\captionsetup[subfigure]{labelformat=empty}
\begin{subfigure}{0.49\textwidth}
	\centering
	\includegraphics[trim = 1cm 6cm 1cm 6.5cm, width=1\textwidth]{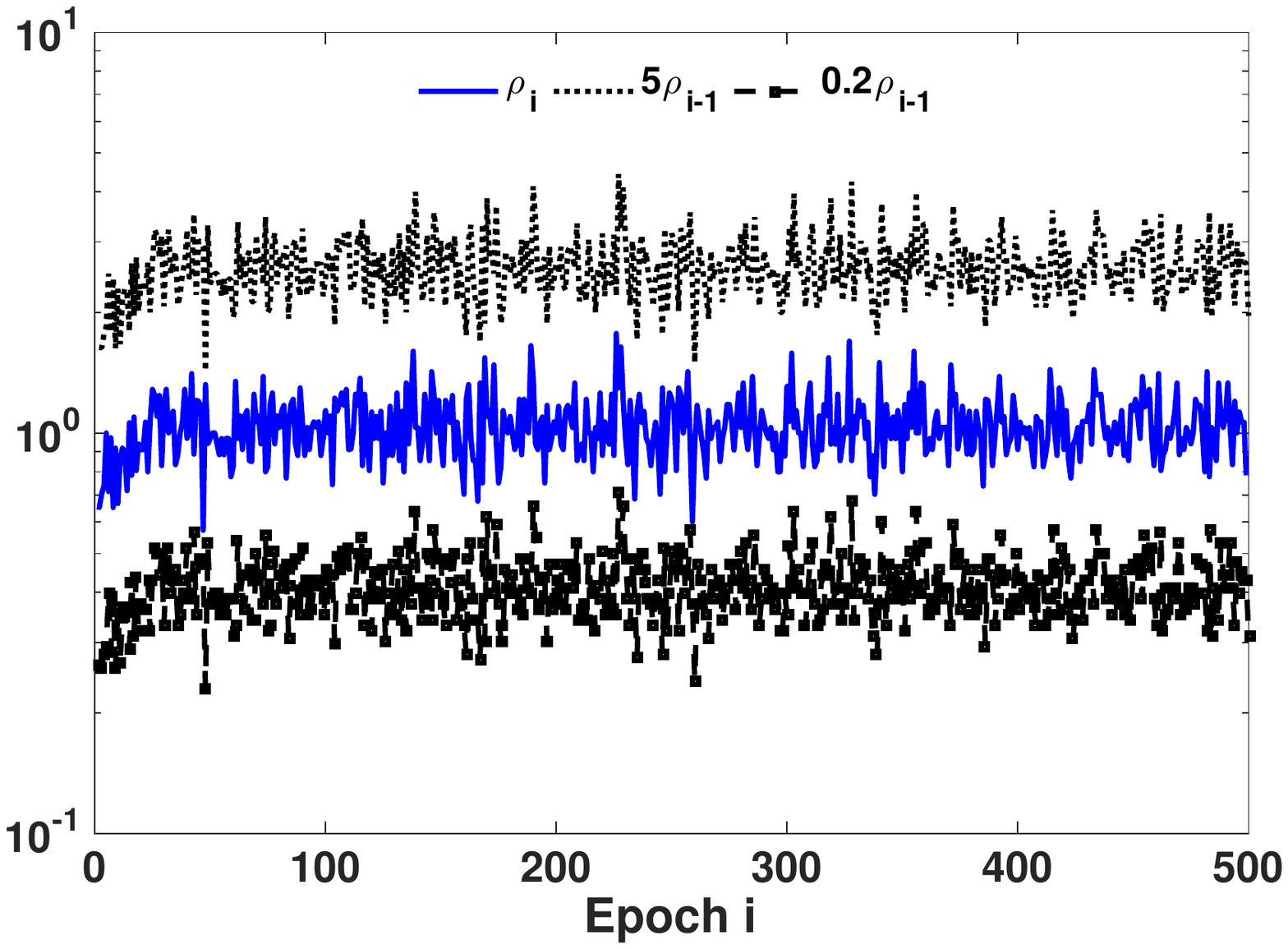} 
	\vspace{-25pt}
	\caption{(a)}
\end{subfigure}
\begin{subfigure}{0.49\textwidth}
	\includegraphics[trim = 1cm 6cm 1cm 6.5cm, width=1\textwidth]{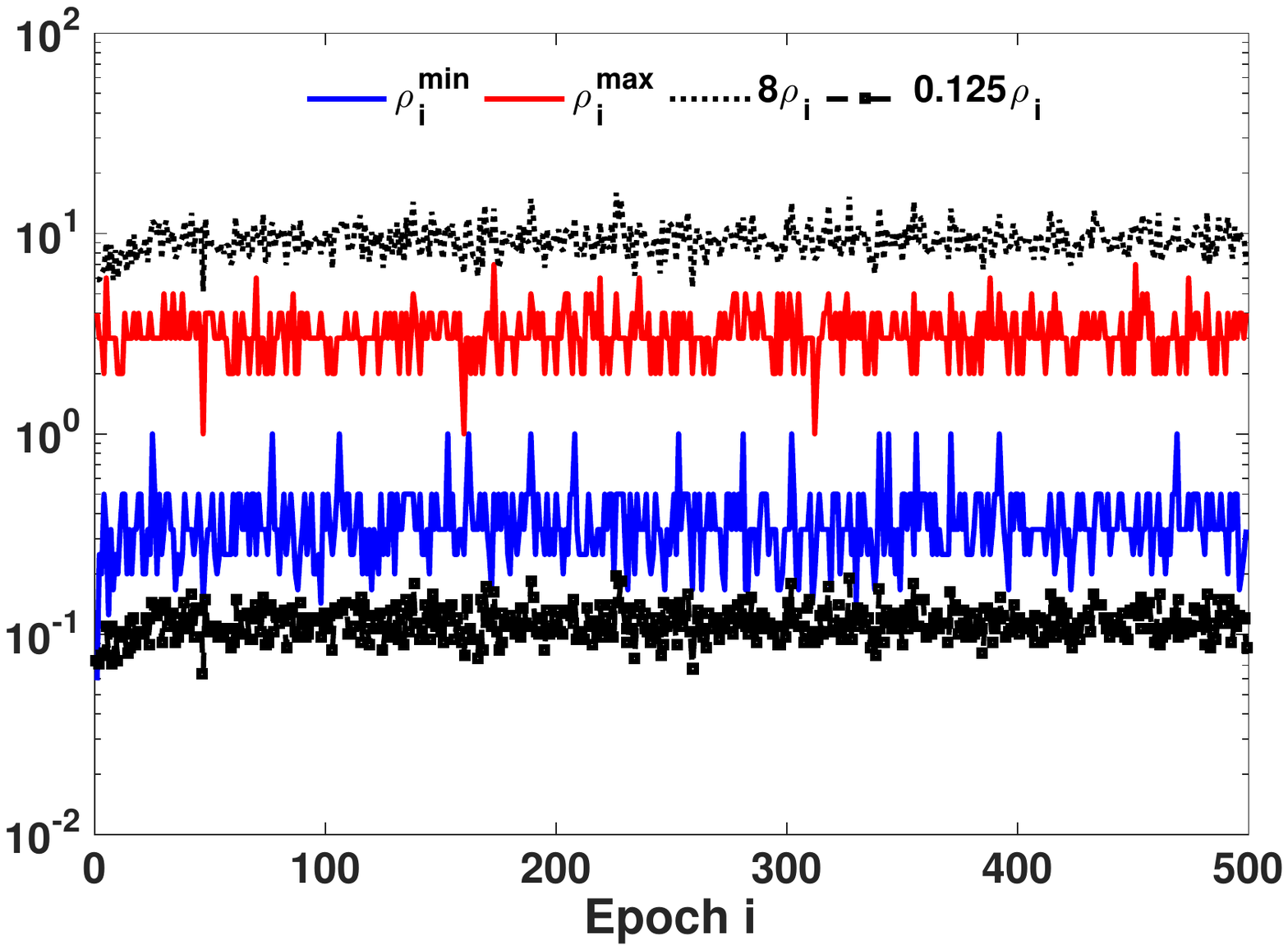}  
	\vspace{-25pt}
	\caption{(b)}
\end{subfigure}
\begin{subfigure}{0.49\textwidth}
	\vspace{-10pt}
	\includegraphics[trim = 1cm 6cm 1cm 6.5cm, width=1\textwidth]{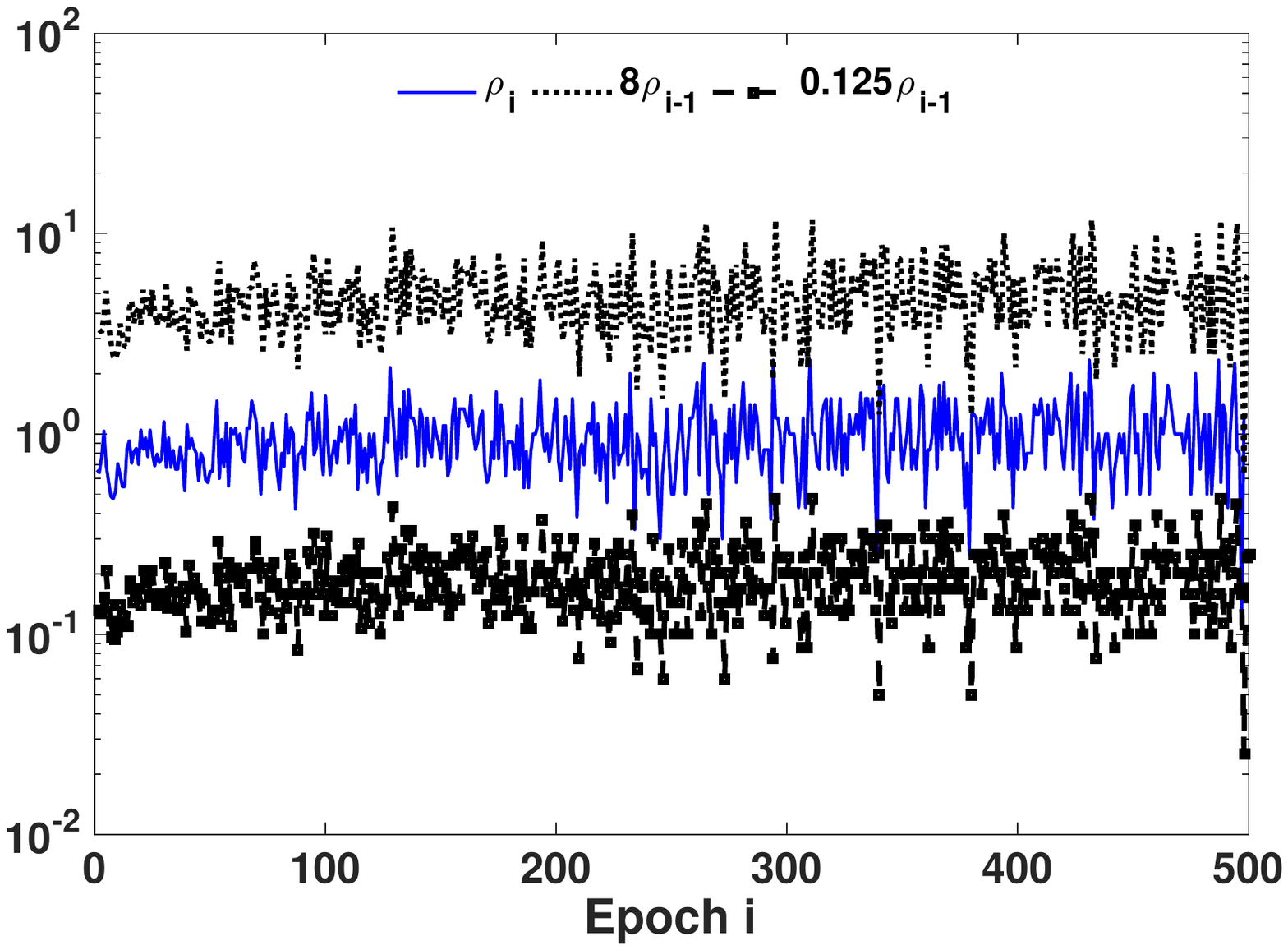}
	\vspace{-25pt}
	\caption{(c)}
\end{subfigure}
\begin{subfigure}{0.49\textwidth}
	\vspace{-10pt}
	\includegraphics[trim = 1cm 6cm 1cm 6.5cm, width=1\textwidth]{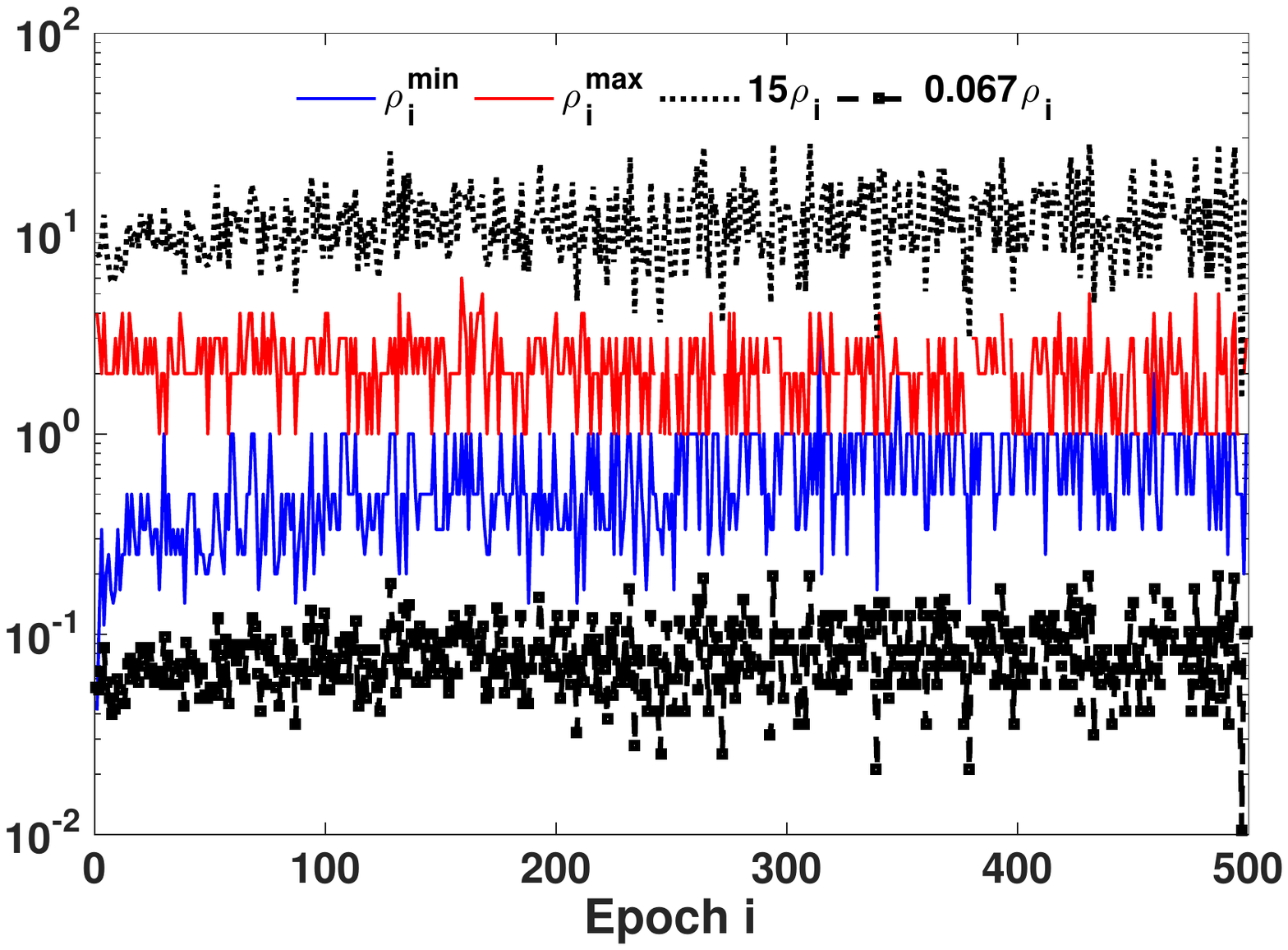} 
	\vspace{-25pt}
	\caption{(d)}
\end{subfigure}
\begin{subfigure}{0.50\textwidth}
	\includegraphics[trim = 1cm 6cm 1cm 6.5cm, width=1\textwidth]{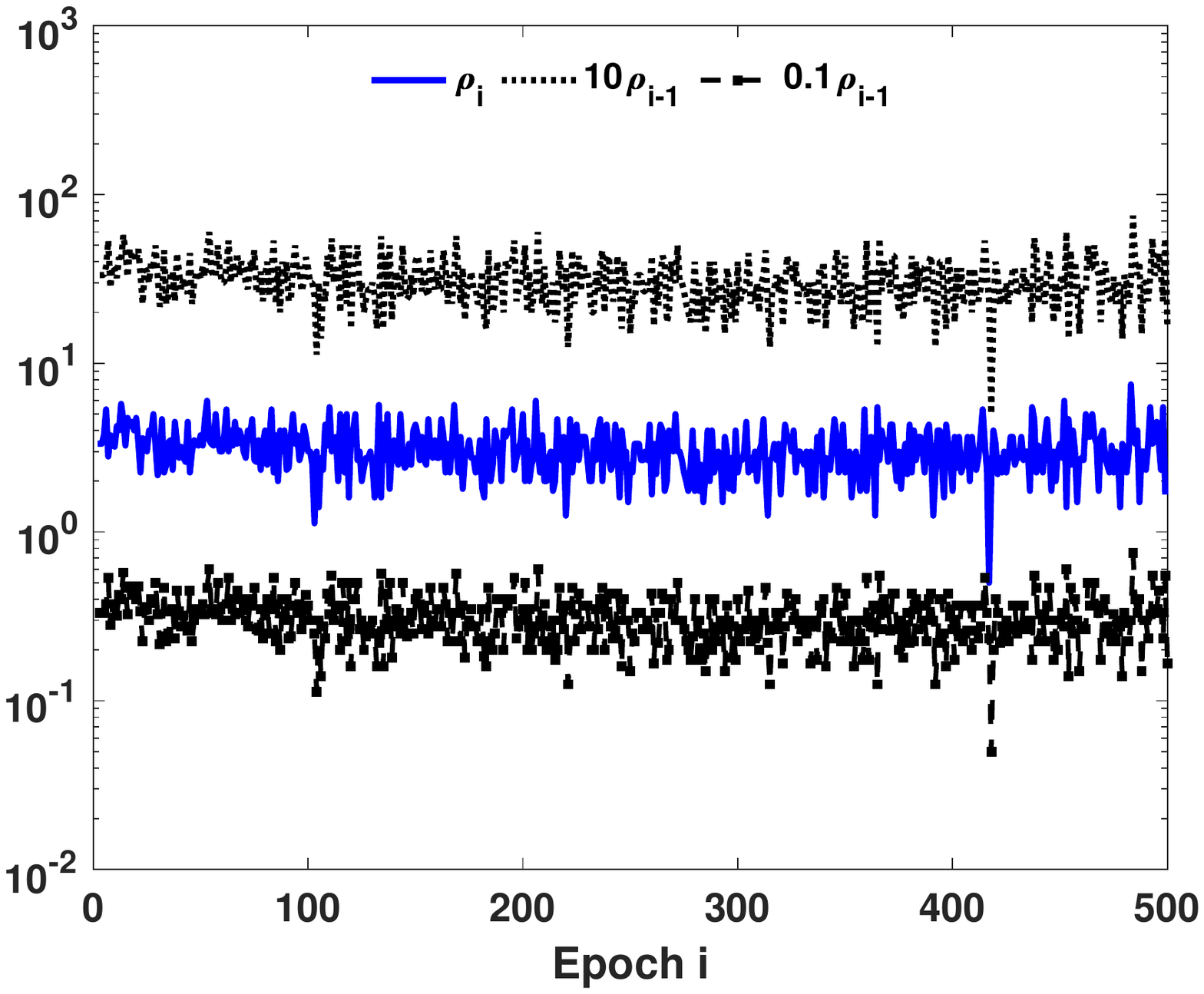} 
	\vspace{-15pt}
	\caption{(e)}
\end{subfigure}
\begin{subfigure}{0.49\textwidth}
	\includegraphics[trim = 1cm 6cm 1cm 6.5cm, width=1\textwidth]{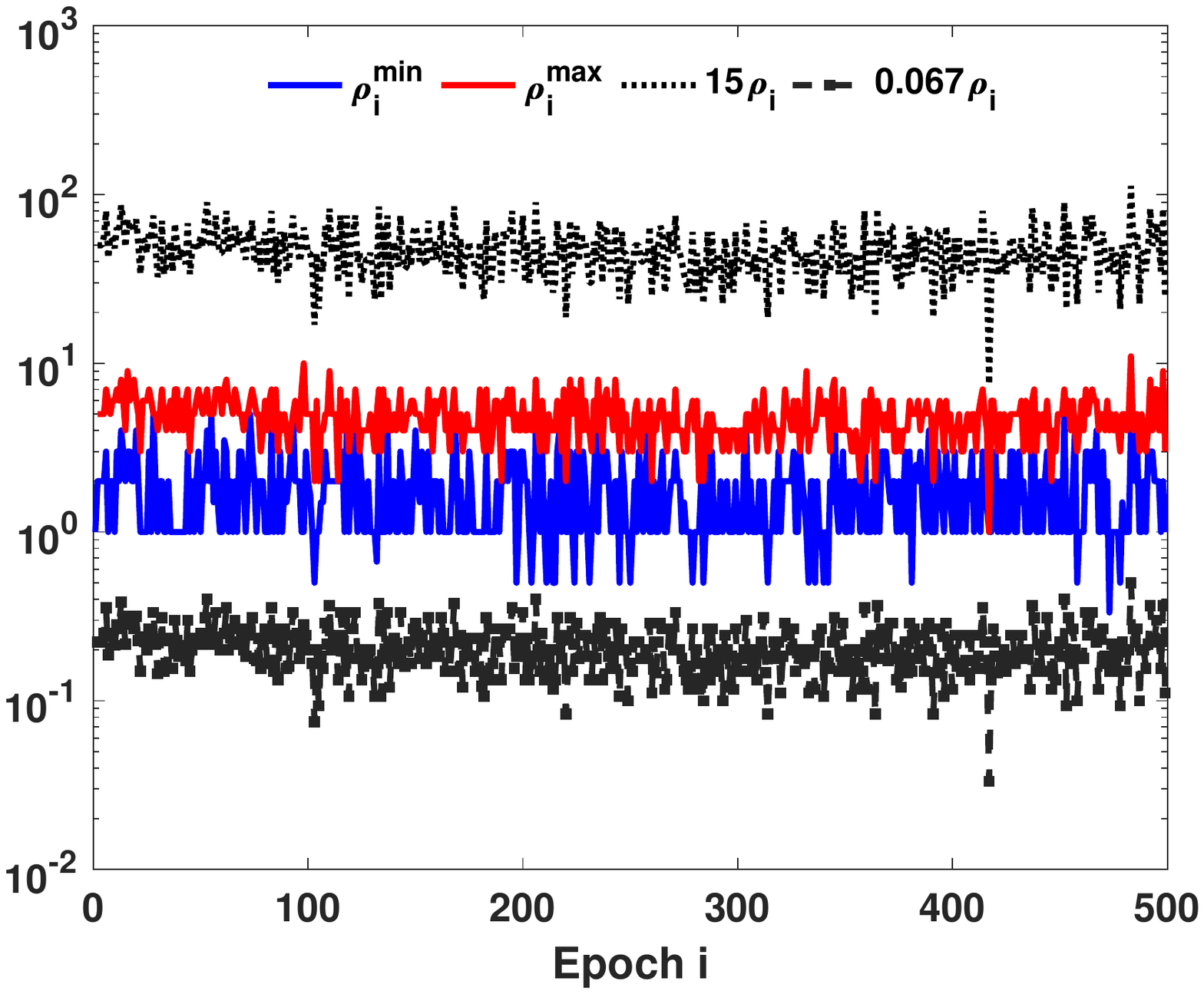} 
	\vspace{-15pt}
	\caption{(f)}
\end{subfigure}
\caption{Testing Assumptions A1 and A2 for: (a) and( b) Debian, (c) and (d) RedHat, (e) and (f) FlatOut. }
\label{fig:assumptionValidation}
\vspace{20pt}
\end{figure*}

To validate our assumptions from Section \ref{sec:modelGMCom}, we performed empirical studies on the join rate of IDs. Recall from Section \ref{sec:modelGMCom}, we define $\rho_i$ to be the rate of join of good IDs in epoch $i$. Since we are interested in only the join rate of good IDs, we assume all IDs that ever join the system during our simulations in this section are good.

\smallskip
{\bf Experimental Setup.} For the Bitcoin network, the system initially consists of 9212 IDs, and the join and departure events are based off the dataset from \cite{neudecker-atc16}. For the BitTorrent networks, we initialize the three networks with 1000 IDs each, and simulate the join and departure events over 1000 epochs. We verified our observations over 20 independent runs for the four networks. 

In order to test if Assumption A1 holds, for each epoch $i \in \{2,3,...,1000\}$, the join rate of epoch $i$, $\rho_i$, was compared to that join rate of the previous epoch, $\rho_{i-1}$. This was done for each of the four networks, the results are plotted in Figure \ref{fig:BitcoinAssumptions}(a) for the Bitcoin network, and Figure \ref{fig:assumptionValidation}(a), (c) and (e) for BitTorrent Debian, BitTorrent RedHat and BitTorrent FlatOut networks, respectively.

Next, we test Assumption A2 for each of the four networks. In particular, we test for each epoch $i \in \{1,2,...,1000\}$, that the join rate between two successive joins by good IDs is within some constant of $\rho_i$. To do so, we measured the minimum and maximum join rate for epoch $i$, denoted by $\rho_i^{min}$ and $\rho_i^{max}$, respectively, and compared these values against $\rho_i$. 

\smallskip

The result  are plotted in Figure \ref{fig:BitcoinAssumptions} (b) for the Bitcoin network, and Figure \ref{fig:assumptionValidation}(b), (d) and (f) for BitTorrent Debian, BitTorrent RedHat and BitTorrent FlatOut Networks, respectively. Table \ref{tab:assumptions} lists the values of constants attached to our assumptions. 

\begin{table}[h!]
\centering
\begin{tabular}{|p{3cm}|P{1.8cm}|P{1.9cm}|P{1.8cm}|P{1.9cm}|}
\hline
\textbf{Network} & $\AOneL$ & $\AOneH$ & $\ATwoL$ & $\ATwoH$ \\
\hline
Bitcoin Network & 0.1 & 10 & 0.0005 & 30 \\
\hline
BitTorrent Debian & 0.2 & 5 & 0.125 & 8 \\
\hline
BitTorrent RedHat & 0.125 & 8 & 0.067 & 15 \\
\hline
BitTorrent FlatOut & 0.1 & 10 & 0.067 & 15 \\
\hline
\end{tabular}	
\caption{Constants Assumptions A1 and A2 for Section \ref{s:JandLAssum}}
\label{tab:assumptions}
\vspace{-20pt}
\end{table}


\subsection{Testing Computational Cost}\label{section:empasym}

We simulate \AlgB and \algGM to validate that they exhibit the predicted spending rate from Theorem \ref{thm:main-upper}.  Specifically, we  validate the asymmetric spending rate $O(\sqrt{\advAveCost\,(\joinRate+1)} + \joinRate )$, where $\advAveCost$ is the adversary's computational cost for solving puzzles divided by the duration of an attack,  and $\joinRate$ is the average join rate of good IDs over the duration of the experiment for the four networks.

\begin{figure*}[t!]
\captionsetup[subfigure]{labelformat=empty}
\begin{subfigure}{0.49\textwidth}
	\centering
	\includegraphics[trim = 1cm 6cm 1cm 6.5cm, width=1\textwidth]{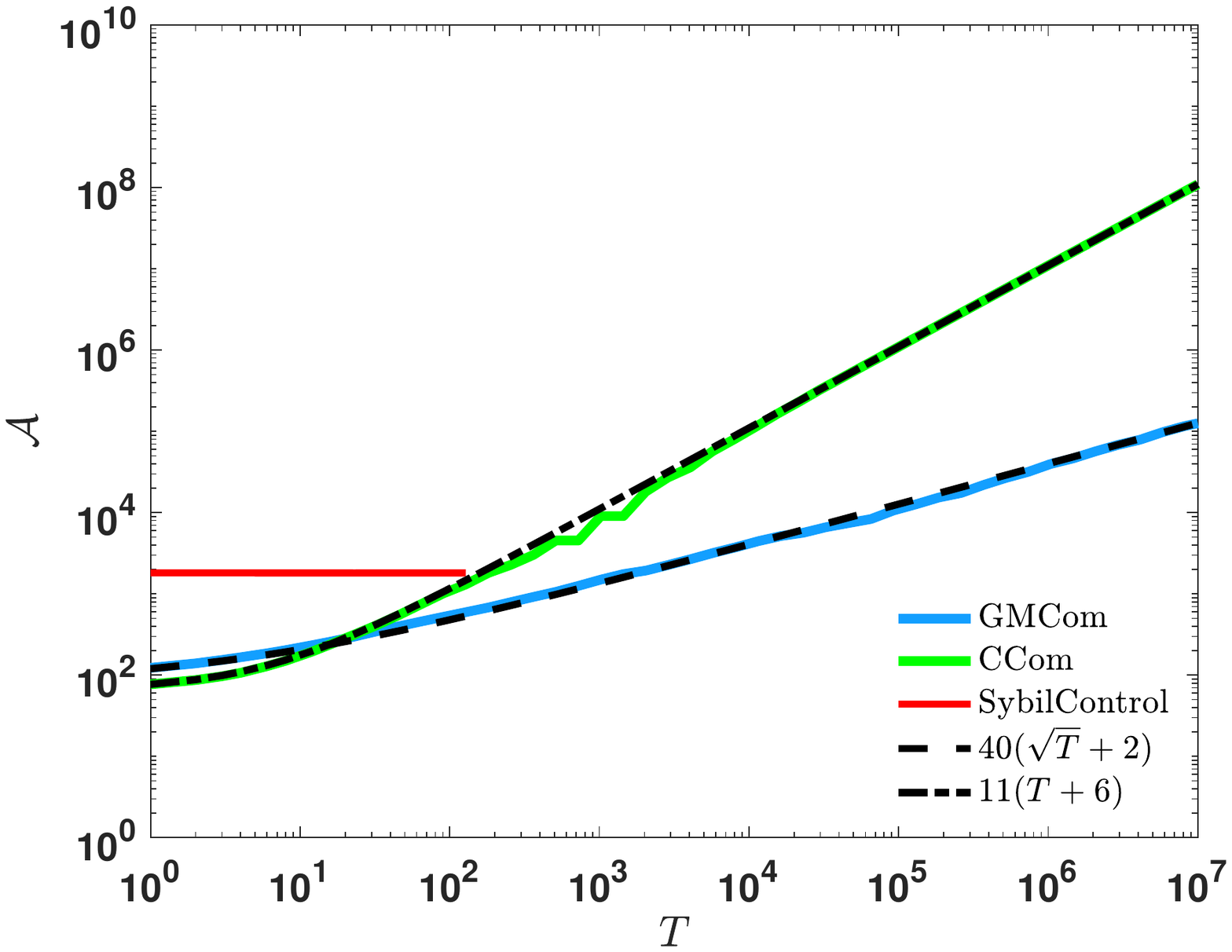} 
	\vspace{-20pt}
	\caption{(a)}
\end{subfigure}
\begin{subfigure}{0.49\textwidth}
	\includegraphics[trim = 1cm 6cm 1cm 6.5cm, width=1\textwidth]{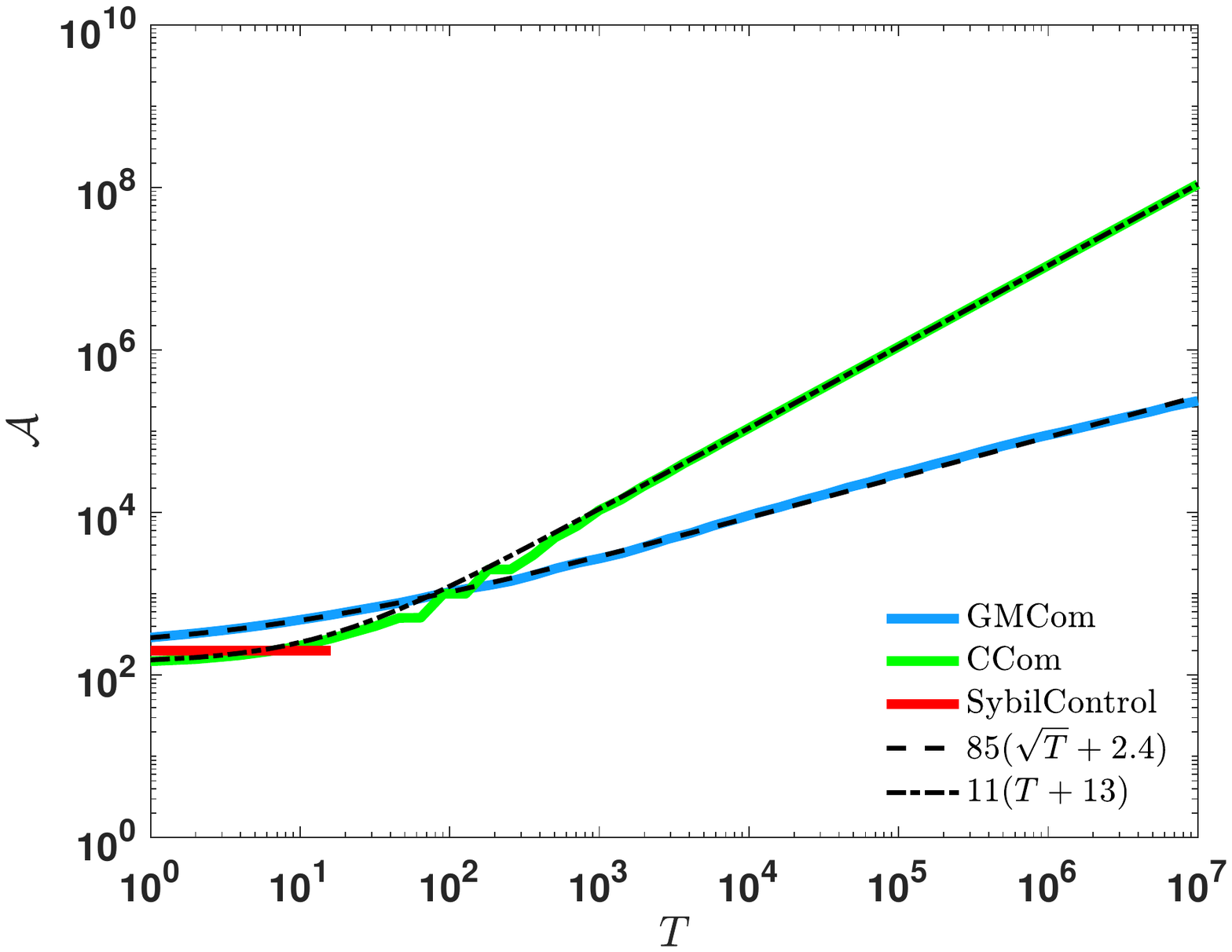}  
	\vspace{-20pt}
	\caption{(b)}
\end{subfigure}
\begin{subfigure}{0.49\textwidth}
	\vspace{-10pt}
	\includegraphics[trim = 1cm 6cm 1cm 6.5cm, width=1\textwidth]{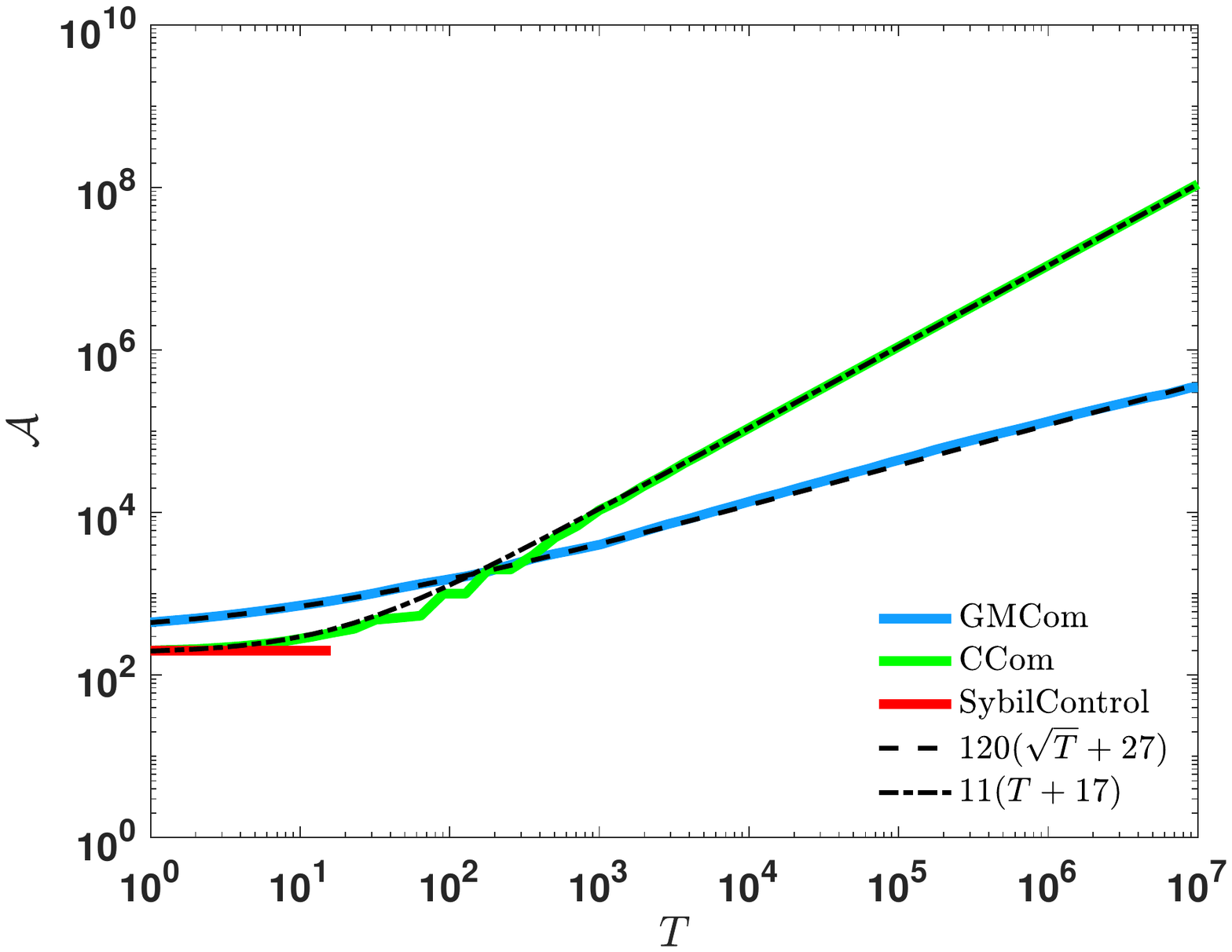} 
	\vspace{-20pt}
	\caption{(d)}
\end{subfigure}
\begin{subfigure}{0.49\textwidth}
	\includegraphics[trim = 1cm 6cm 1cm 6.5cm, width=1\textwidth]{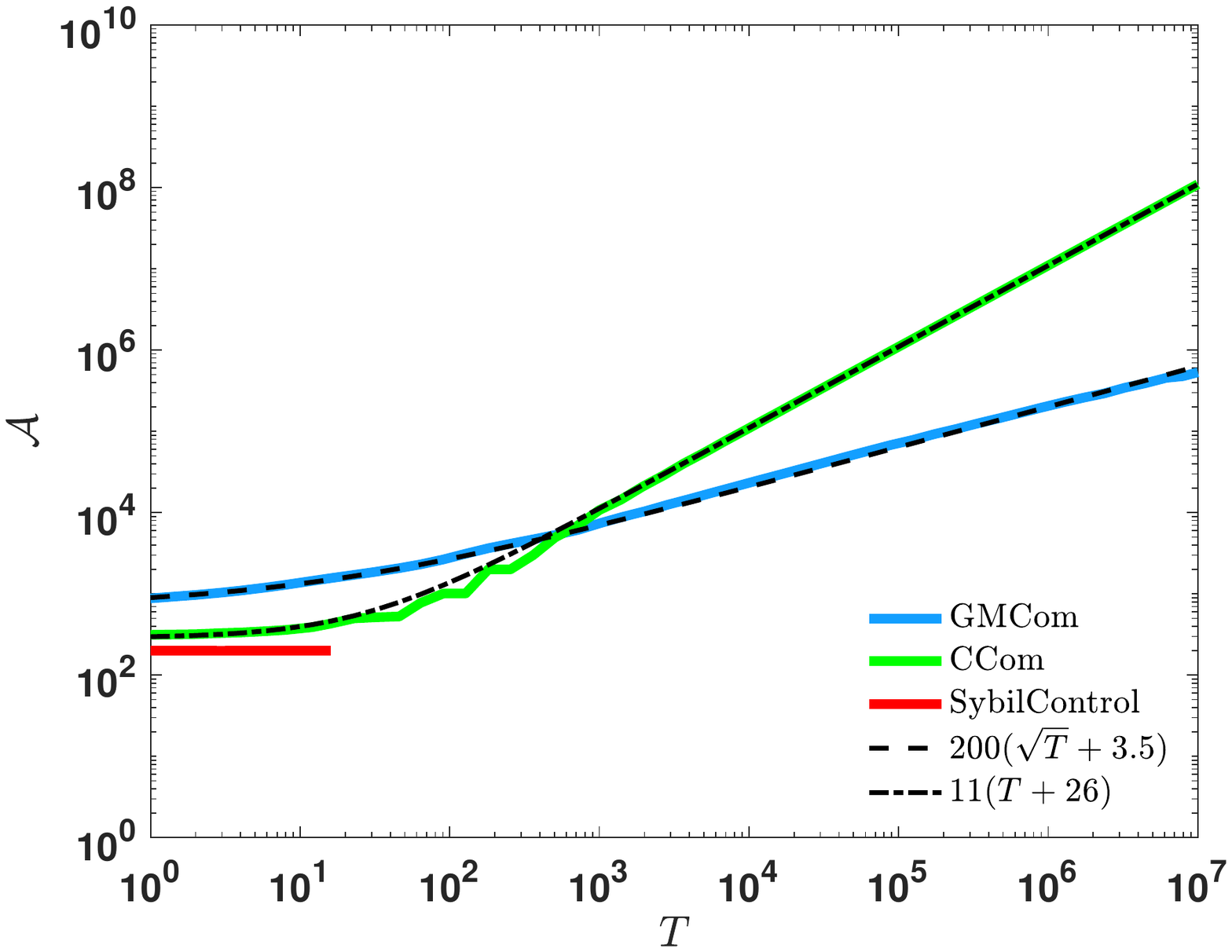}
	\vspace{-20pt}
	\caption{(c)}
\end{subfigure}
\vspace{-5pt}
\caption{Comparison of algorithmic cost to adversarial cost for \algGM, \AlgB and SybilControl.}
\label{fig:AvsT}
\end{figure*} 

We compare the performance of our algorithms against another PoW-based Sybil defense,  \textbf{\AlgA}~\cite{li:sybilcontrol}. Under this algorithm, each ID must solve a puzzle to join the system. Additionally, each ID periodically tests its neighbors with a puzzle, removing from its list of neighbors those IDs that fail to provide a solution within a time limit;  these tests are not coordinated between IDs. An ID may be a neighbor to more than one ID and so receive multiple puzzles; these are combined into a single puzzle whose solution satisfies all the received puzzles. 

Similar to Section \ref{s:JandLAssum}, the Bitcoin network initially consists of $9212$ IDs and the join and departure events are scheduled based on the events provided in the dataset. Also, for the three BitTorrent network, we initialize the system with 1000 IDs and the join/ departure events are scheduled based on the parameters and distribution from \cite{neudecker-atc16}. 

We assume $\alpha = 1/14$, and  $T$ ranges over $[2^0,2^{100}]$, where for each  value of $T$, the system is simulated for $10,000$ seconds. The adversary solves entrance puzzles to add bad IDs to the system. We pessimistically ignore the cost paid by the adversary for solving purge puzzles. We average our results over 20 independent runs for each of the four networks for the three algorithms. To avoid cluttering the plots, we omit error bars as they are negligible.

Figure \ref{fig:AvsT} illustrates our results. The blue line depicts the average computational cost to good IDs per second obtained by executing \algGM when the adversary spends $T$ per second. The green and red lines are the plots for \AlgB and \AlgA, respectively.

Note that the $x$-axis and $y$-axis are both $\log$ scaled. Initially, the blue line increases slowly due to $T$ not being substantially larger than $\joinRate$, and hidden constants. However, as $T$ grows, we observe behavior very close to $G = \sqrt{T}$, which validates the asymptotic behavior of the asymmetric cost. We cut off the red line at the point at which \AlgA~ is no longer able to ensure that the fraction of bad IDs is less than $1/2$.

\newpage
\subsection{Proposed Heuristics}\label{section:heuristics}

\begin{wrapfigure}{r}{0.55\textwidth}
	\centering
	\includegraphics[trim = 1cm 8cm 1cm 10cm, width=.48\textwidth]{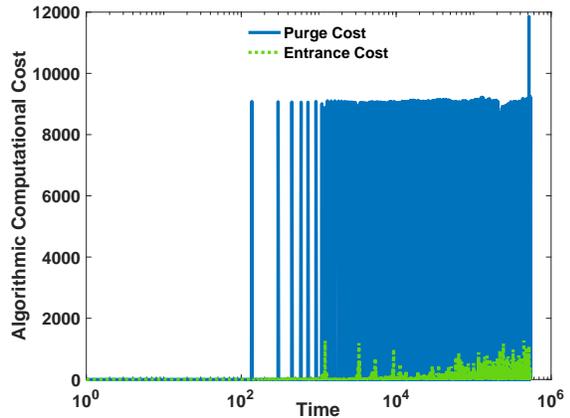}
	\caption{Comparison of purge cost to entrance cost for Bitcoin Network.}
	\vspace{-5pt}
	\label{fig:purgeVSentrance}
\end{wrapfigure}

In this section, we explore heuristics to improve performance of our algorithm. To determine effective heuristics, we first study two separate costs: the purge cost and the entrance cost to the good IDs. We studied these two costs for the Bitcoin Network in the absence of an attack. 

In Figure \ref{fig:purgeVSentrance}, we plot the costs as a function of time. In order to better visualize the frequency of purges, the $x$-axis is a log plot. It is clear that our algorithmic cost is dominated by the cost of purges. Hence, we focus on reducing the frequency of purges, without compromising correctness. We describe three heuristics below, followed by an evaluation of their performance.

\smallskip

\noindent\textbf{Heuristic 1:} We use the symmetric difference to determine when to do a purge.  Specifically, for iteration $i$, if $|(\curIDs \cup \setIDs_{i-1}) - (\curIDs \cap \setIDs_{i-1})| \geq |\setIDs_{i-1}|/11$, then a purge is executed.  This ensures that the fraction of bad IDs can increase by no more than in our original specification.  But it potentially lengthens the amount of time between purges, for example, in the case when some ID joins and departs repeatedly.

\smallskip

\begin{figure*}[b!]
\captionsetup[subfigure]{labelformat=empty}
\begin{subfigure}{0.49\textwidth}
	\centering
	\includegraphics[trim = 1cm 6cm 1cm 6.5cm, width=1\textwidth]{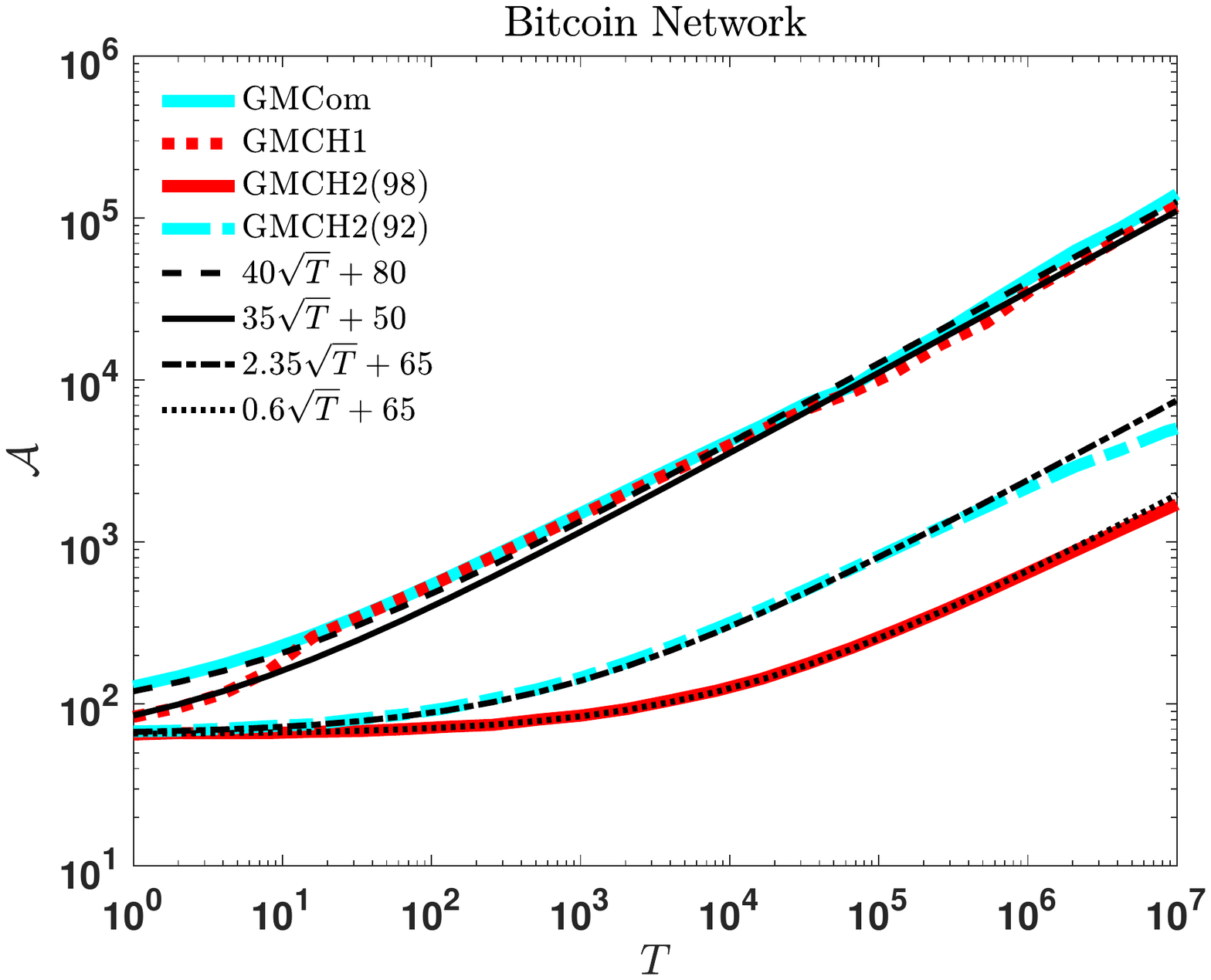} 
	\vspace{-20pt}
	\caption{(a)}
\end{subfigure}
\begin{subfigure}{0.49\textwidth}
	\includegraphics[trim = 1cm 6cm 1cm 6.5cm, width=1\textwidth]{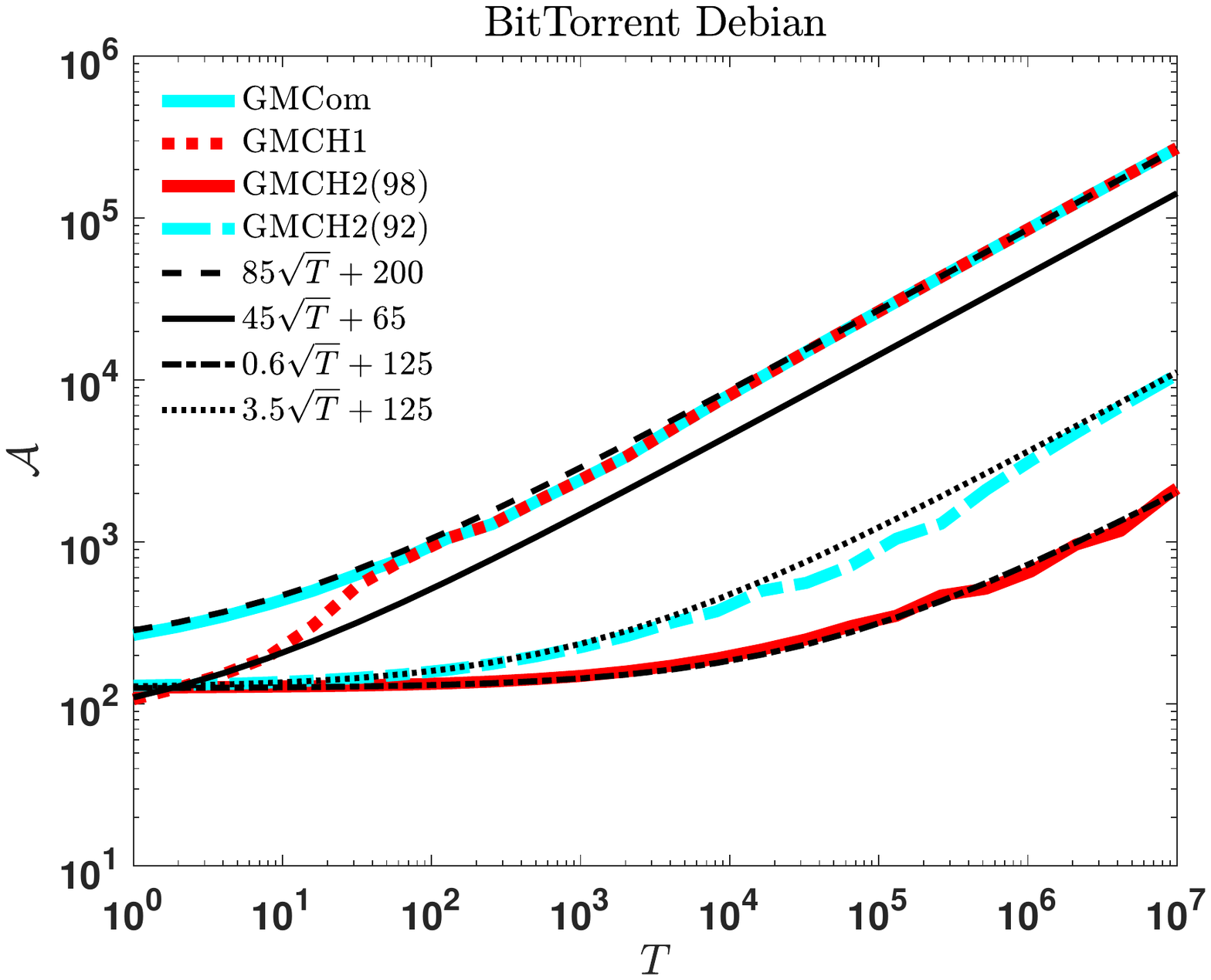}  
	\vspace{-20pt}
	\caption{(b)}
\end{subfigure}
\begin{subfigure}{0.49\textwidth}
	\includegraphics[trim = 1cm 6cm 1cm 6.5cm, width=1\textwidth]{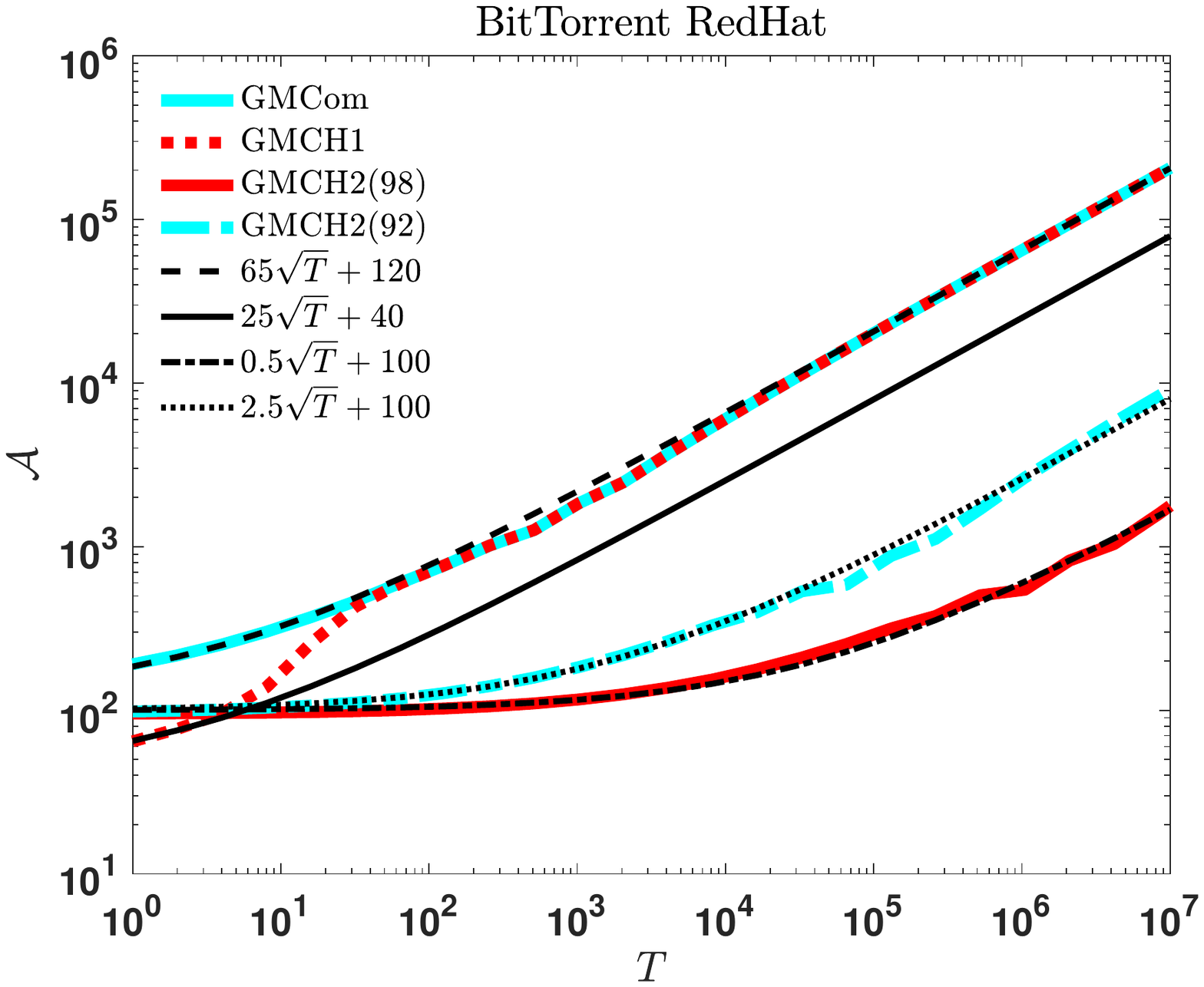}
	\vspace{-20pt}
	\caption{(c)}
\end{subfigure}
\begin{subfigure}{0.49\textwidth}
	\includegraphics[trim = 1cm 6cm 1cm 6.5cm, width=1\textwidth]{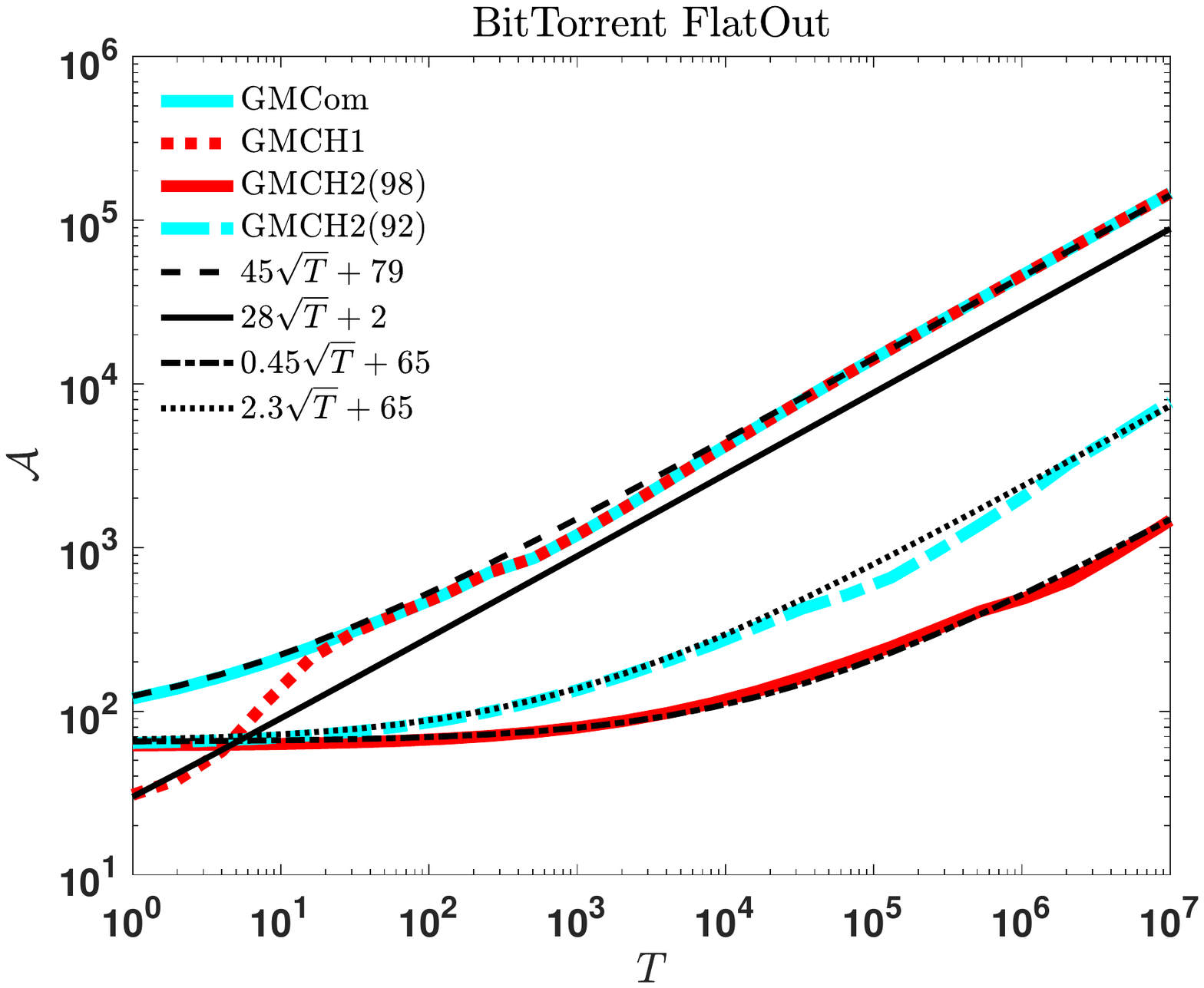} 
	\vspace{-20pt}
	\caption{(d)}
\end{subfigure}
\caption{Comparison of algorithmic cost to adversarial cost for \algGM and proposed heuristics.}
\vspace{-10pt}
\label{fig:Heuristic}
\end{figure*}

\noindent\textbf{Heuristic 2:} We use the estimated good join rate to bound the maximum number of bad IDs that might have joined during the current iteration. This allows us to upper-bound the fraction of bad IDs in the system and purge only when the system's correctness may be at risk. We use this fact in conjunction with Heuristic 1 to safely delay the purge.

\smallskip

\noindent\textbf{Heuristic 3:} We propose the use of existing Sybil-detection techniques to remove some Sybil IDs from the system.  For example, recent works have  explored the possibility of identifying Sybil IDs based on the network topology \cite{yu2008sybillimit,danezis2009sybilinfer,misra2016sybilexposer}. In our experiments, we focus on SybilInfer from \cite{misra2016sybilexposer}, and use it as a black box to possibly detect and discard Sybil IDs as they enter the system. The accuracy of the blackbox in detecting honest and Sybil IDs is a parameter in our experiments.  We focus on two specific accuracy values, $0.98$ and $0.92$, based on emperical results in \cite{misra2016sybilexposer}.

\smallskip

To evaluate the performance of these heuristics, we use \algGM as a baseline. The experimental setup is the same as Section~\ref{section:empasym}, with the same four data sets. We define \defn{GMCH1} to be \algGM using both Heuristic 1 and Heuristic 2. We define \defn{GMCH2(98)} and \defn{GMCH2(92)} to be \algGM using Heuristics 1 and 2, and also Heuristic 3, with the accuracy parameter of Heuristic 3 as 0.98 and 0.92.

Figure \ref{fig:Heuristic} illustrates our results. These indicate that the GMCH1 is effective in reducing costs when adversarial effort is limited. On the other hand, GMCH2 reduces costs significantly during adversarial attack, with improvements of up to two orders of magnitude during the most significant attacks tested.



\section{Conclusion and Future Work}~\label{sec:future}

We have presented and analyzed two PoW-based algorithms for defending against the Sybil attack. Our empirical work suggests that these algorithms are competitive with the contemporary \POW-based Sybil defense, \AlgA, and that they can be significantly more efficient when the system is under  attack.  Finally, we have proved a lower bound showing that our algorithm's computational cost is asymptotically optimal among a large class of Sybil-defense algorithms.   

Many open problems remain including the following. One of particular interest is: Can we adapt our technique to secure multi-party computation? The problem of \defn{secure multi-party computation (MPC)} involves designing an algorithm for the purpose of computing the value of an $n$-ary function $f$ over private inputs from $n$  IDs $x_1, x_2,...,x_n$, such that the IDs learn the value of $f(x_1,x_2,...,x_n)$, but learn nothing more about the inputs than what can be inferred from this output of $f$.  The problem is generally complicated by the assumption that an adversary controls a hidden subset of the IDs.  In recent years, a number of attempts have been made to solve this problem for very large $n$ manner~\cite{applebaum2010secrecy,beerliova2006efficient,bogetoft2009secure,damgard2006scalable,damgaard2008scalable,dani2014quorums,goldreich1998secure}.  We believe that our technique for forming committees in a dynamic network could be helpful to solve a dynamic version of this problem, while ensuring that the resource costs to the good IDs is commensurate with the resource costs to an adversary.  Secure MPC may have applications in the area of smart contracts for cryptocurrencies~\cite{kosba2016hawk,benhamouda2019supporting,christidis2016blockchains}.

Another avenue of future work considers scenarios where ``good'' IDs may not always blindly follow our algorithm.  Instead, these IDs may be rational but selfish, in that they seek to optimize some known utility function.  The adversary still behaves in a worst-case manner, capturing the fact that the utility function of the adversary may be completely unknown.  This approach is similar to BAR (Byzantine, Altruistic, Rational) games~\cite{clement:theory} in distributed computing, see also~\cite{abraham:distributed,Aiyer:2005:BFT,Vilaca:NBT,Wong:2010:MBA}.  Extending our result to accommodate this model is of interest.

Finally, there is a substantial body of literature on attack-resistant overlays; for example~\cite{fiat:making,awerbuch:towards,awerbuch:random,awerbuch:towards2,young:practical,saia:reducing,young:towards,guerraoui:highly,naor:novel,naor_wieder:a_simple,sen:commensal,saad:self-healing,saad:self-healing2,awerbuch_scheideler:group,AnceaumeLRB08,HK}. These results critically depend on the fraction of bad IDs  always being upper-bounded by a constant less than $1/2$. However, there are only a handful of results that propose a method for guaranteeing this bound. A natural idea is to examine whether our algorithms can be applied to this setting in order to guarantee this bound, via the \sgoal.

\end{document}